\documentclass[a4paper,11pt]{article}                 

\pdfoutput=1

\makeatletter
\renewcommand*\@fnsymbol[1]{\the#1}
\makeatother

\usepackage[hmargin=3.5cm,vmargin=4cm]{geometry}
\usepackage{graphicx}
\usepackage{mathptmx}  
\DeclareMathAlphabet{\mathcal}{OMS}{cmsy}{m}{n}  
\usepackage{amsmath}
\usepackage{amssymb}
\usepackage{amsthm}
\usepackage{thmtools}
\usepackage{mathtools}
\usepackage{amsfonts}
\usepackage{titlesec}
\usepackage[hidelinks]{hyperref}
\usepackage[british]{babel}
\usepackage{paralist}
\usepackage{enumitem}
\usepackage{subcaption} 
\usepackage{wasysym}
\usepackage{xcolor}

\setlist[itemize]{noitemsep}
\setlist[enumerate]{noitemsep}

\titleformat{\subsection}{\normalfont\normalsize}{\thesubsection}{.5em}{}
\titleformat{\subsubsection}{\normalfont\normalsize\itshape}{\thesubsubsection}{.5em}{}
\titlespacing*{\section}
{0pt}{5ex plus 1ex minus .1ex}{3ex plus .2ex minus .1ex}
\titlespacing*{\subsection}
{0pt}{5ex plus 1ex minus .1ex}{3ex plus .2ex minus .1ex}
\titlespacing*{\subsubsection}
{0pt}{5ex plus 1ex minus .1ex}{3ex plus .2ex minus .1ex}

\newtheoremstyle{dotlessTHM}{2ex}{1ex}{\itshape}{}{\bfseries}{}{.5em}{\thmname{#1}\thmnumber{ #2}\thmnote{ (#3)}}
\newtheoremstyle{dotlessDEF}{2ex}{1ex}{}{}{\bfseries}{}{.5em}{\thmname{#1}\thmnumber{ #2}\thmnote{ (#3)}}
\newtheoremstyle{dotlessREM}{2ex}{1ex}{}{}{\itshape}{}{.5em}{\thmname{#1}\thmnumber{ #2}\thmnote{ (#3)}}

\declaretheorem[name=Theorem,numberwithin=section,style=dotlessTHM]{theorem}
\declaretheorem[name=Definition,sharenumber = theorem, style = dotlessDEF]{definition}
\declaretheorem[name=Corollary,sharenumber = theorem, style = dotlessTHM]{corollary}
\declaretheorem[name=Proposition,sharenumber = theorem, style = dotlessTHM]{proposition}
\declaretheorem[name=Lemma,sharenumber = theorem, style = dotlessTHM]{lemma}
\declaretheorem[name=Remark,sharenumber = theorem, style = dotlessREM]{remark}
\declaretheorem[name=Example,sharenumber = theorem, style = dotlessREM]{example}

\numberwithin{equation}{section}
\numberwithin{figure}{section}

\makeatletter
\let\@addpunct\@gobble
\makeatother

\newcommand{\risk}[3]{\rho_{\mathcal{#1},#2}(#3)}
\newcommand{\riskR}[3]{R_{\mathcal{#1},#2}(#3)}
\newcommand{\riskA}[3]{\rho_{#1,#2}(#3)}
\newcommand{\riskAR}[3]{R_{#1,#2}(#3)}

\newcommand{\riskset}[2]{ \inf\left\{ \lambda \in [0,1] \,\big|\, (1-\lambda){#1}_T + \lambda \tfrac{{#1}_0}{{#2}_0} {#2}_T \in \mathcal{A} \right\}  } 

\newcommand{\fspace}[0]{\mathbb{R}_{>0} \times \mathcal{X}}
\newcommand{\aspace}[0]{\mathbb{R}_{>0} \times \mathcal{A}}
\newcommand{\acc}[0]{\mathcal{A}}

\title{Intrinsic risk measures}
\author{\normalsize Walter Farkas\thanks{\href{mailto:walter.farkas@bf.uzh.ch}{walter.farkas@bf.uzh.ch}}\; and Alexander Smirnow\thanks{\href{mailto:alexander.smirnow@uzh.ch}{alexander.smirnow@uzh.ch}}\\
\normalsize Department of Banking and Finance, University of Zurich, \\
\normalsize Department of Mathematics, ETH Z\"urich  }
\date{\normalsize First version: October 27, 2016}

\begin{document}
\maketitle

\begin{abstract}
\noindent Monetary risk measures are usually interpreted as the smallest amount of external capital that must be added to a financial position to make it acceptable.
We propose a new concept: intrinsic risk measures and argue that this approach provides a direct path from unacceptable positions towards the acceptance set.
Intrinsic risk measures use only internal resources and return the smallest percentage of the currently held financial position which has to be sold and reinvested into an eligible asset such that the resulting position becomes acceptable. 
While avoiding the problem of infinite values, intrinsic risk measures allow a free choice of the eligible asset and they preserve desired properties such as monotonicity and quasi-convexity.
A dual representation on convex acceptance sets is derived and the link of intrinsic risk measures to their monetary counterparts on cones is detailed. \\




\end{abstract}

\section{Introduction}
\label{sec:intro}
Risk measures associated with acceptance criteria as introduced by P.~Artzner, F.~Delbaen, J.~Eber, and D.~Heath \cite{bib:ADHE} are maps $\rho_{\mathcal{A},r}$ 
from a certain function space $\mathcal{X}$ to $\mathbb{R}$ of the form
\begin{align} \label{Eq:MonRisk_constElig}
\risk{A}{r}{X_T} = \inf \left\{ m \in \mathbb{R} \,|\, X_T + m r \mathbf{1}_\Omega \in \mathcal{A} \right\} \,.
\end{align}
In words, the risk of a financial position $X_T \in \mathcal{X}$ is measured by the smallest amount of external money $m \in \mathbb{R}$ which must be invested into a risk-free reference instrument with constant return rate $r > 0$ in order to make it acceptable, that is, to make it an element of an acceptance set $\acc \subset \mathcal{X}$.
In this approach, acceptance sets form the primary objects and an associated risk measure is given by the distance between financial position and the boundary of the acceptance set with respect 
to the direction~$r \mathbf{1}_\Omega$. 
More recently, this approach has been re-linked to the original idea of using eligible assets with random return rate $r : \Omega \rightarrow \mathbb{R}_{>0}$ by P.~Artzner, F.~Delbaen, and P.~Koch-Medina \cite{bib:ADKM} and \mbox{D.~Konstantinides} and C.~Kountzakis \cite{bib:KK} among others, or has even been extended to processes by \mbox{M.~Fritelli} and G.~Scandolo \cite{bib:FrittelliScandolo}.
W.~Farkas, P.~Koch-Medina, and C.~Munari in \cite{bib:WPC2} and \cite{bib:WPC1} 
proposed to investigate general eligible assets $r:\Omega \rightarrow \mathbb{R}_{\geq 0}$ and acceptance sets, revealing significant shortcomings of the simplified constant approach and pointing out the close interplay between eligible assets and acceptance sets.
They work with a traded asset $S = (S_0,S_T)$ defined by its initial unitary price $S_0 \in \mathbb{R}_{>0}$ and its random payoff \mbox{$S_T:\Omega \rightarrow \mathbb{R}_{\geq 0}$} and replace~$r$ in \hyperref[Eq:MonRisk_constElig]{Equation (\ref*{Eq:MonRisk_constElig})} by the random return $\frac{S_T}{S_0}$ to arrive at an extended definition
\begin{align} \label{Eq:MonetaryRisk_randElig}
\risk{A}{S}{X_T} = \inf \left\{ m \in \mathbb{R} \,\big|\, X_T + \tfrac{m}{S_0} S_T \in \mathcal{A} \right\}.
\end{align}
The term $\frac{m}{S_0} S_T$ is interpreted as the payoff of $\frac{m}{S_0}$ units of asset $S$.
Consequently, written as $\risk{A}{S}{X_T}/S_0$, this risk measure can also be thought of as the smallest number of units of $S$ that need to be bought and added to the position $X_T$ to make it acceptable.

The choice $S = (1,r\mathbf{1}_\Omega)$ highlights that \hyperref[Eq:MonRisk_constElig]{Equation~(\ref*{Eq:MonRisk_constElig})} is a special case of \hyperref[Eq:MonetaryRisk_randElig]{Equation~(\ref*{Eq:MonetaryRisk_randElig})}. 
If $S_T$ is bounded away from zero, meaning that for some $\varepsilon > 0$ the inequality $S_T \geq \varepsilon$ holds ($\mathbb{P}$-a.s.),
a reduction of the latter to the initial definition is immediate and constitutes the basis for the simplified approach with constant return.
Unfortunately, the assumption that the payoff $S_T$ is bounded away from zero excludes relevant financial instruments such as defaultable bonds or options from the set of possible eligible assets. Moreover, the reduction can lead to alterations of the structure imposed on the acceptance set.
Consequently, allowing eligible assets with payoffs which are not necessarily bounded away from zero is a key point in the analysis of several concrete financial situations.

\vspace{8pt}
However, regardless of the particular definition of risk measure, these approaches are in line with the point of view in \cite[Section 2.1]{bib:ADHE}: 
\begin{quote}
`\emph{The current cost of getting enough of this or these [commonly accepted] instrument(s) is a good candidate for a measure of risk of the initially unacceptable position.}'
\end{quote}
This seminal idea does not only allow one to rank financial positions according to their risk, but also suggests a procedure to make an unacceptable position acceptable.
Referring to cash-additivity (Axiom T in \cite{bib:ADHE}), P.~Artzner, F.~Delbaen, J.~Eber, and D.~Heath claim in \cite[Remark 2.7]{bib:ADHE} that
\begin{quote}
`\emph{By insisting on references to cash and to time, [...] our approach goes much further than the interpretation [...] that ``the main function of a risk measure is to properly rank risks."}'
\end{quote}
However, in order to truly go beyond ranking risks and to apply this procedure, one must carry or raise the monetary amount $\rho_{\acc,S}(X_T)$. 
This raises the question as to which extent this method is applicable and how the acquisition of additional capital can be incorporated into the risk measure.

One possible approach is to sell part of the financial position to raise capital and invest it in the eligible asset, as was already mentioned in \cite[Section 2.1]{bib:ADHE}:
\begin{quote}
`\emph{For an unacceptable risk [...] one remedy may be to alter the position.}'
\end{quote}
The aim of our work is to develop this thought towards a new class of risk measures, which we will call \emph{intrinsic risk measures}.
Traditional risk measures are defined via (hypothetical) additional capital, which is not always available in reality.
Thus, we propose a class of risk measures that only allows the usage of internal capital contained in the financial position.
This new concept leads to a significant change of mentality.
It extends the scope of applications and eliminates problems of infinite values.
It also requires one to devote attention to the initial value of a position and its interplay with the desired eligible asset.

We develop our approach based on acceptance sets $\mathcal{A} \subset \mathcal{X}$ as primary objects and on the extended framework of general eligible assets $S = (S_0, S_T ) \in \mathbb{R}_{>0} \times \mathcal{A}$.

The intrinsic risk of a financial position of interest $X = (X_0,X_T )$ defined by its initial value $X_0 \in \mathbb{R}_{>0}$ and future payoff or net worth $X_T \in \mathcal{X}$ is given by 
\begin{equation} \label{Eq:FirstMentioningOfIntrinsicRisk}
\riskR{A}{S}{X} = \riskset{X}{S}\,.
\end{equation}
The intrinsic risk measure returns the smallest percentage of a given position that needs to be sold and reinvested into the eligible asset $S$, at inception, such that the resulting position is deemed acceptable.

By selling part of the position, required capital is raised and reinvested, resulting in a convex combination of two random variables.
This approach suggests a new way to shift unacceptable positions towards the acceptance set.
In particular, the treatment of risk measures taking infinite values and losing their operational applicability becomes superfluous.
Furthermore, standard properties such as monotonicity and quasi-convexity are preserved and they can be imposed using just the structure of the underlying acceptance set.
\\

The subsequent work has grown from the master's thesis of A.~Smirnow \cite{bib:AS} and is structured as follows. 
In \hyperref[Section:Prelim]{Section \ref*{Section:Prelim}}, the notion of acceptance sets and traditional risk measures are introduced, linking these two concepts and reviewing important properties.
The aim is to give a short overview of the advancements in risk measure theory and to lay the foundation for the intrinsic risk measure. 
In \hyperref[Sec:IntrinsicRisk]{Section \ref*{Sec:IntrinsicRisk}}, we define intrinsic risk measures and derive basic properties in juxtaposition with traditional risk measures.
In \hyperref[Sec:IntrinsicRiskonPosHom&Cashadditive]{Section \ref*{Sec:IntrinsicRiskonPosHom&Cashadditive}}, under the assumption of conic acceptance sets we associate intrinsic risk measures with traditional risk measures and we show that they can be expressed as functions of one another.
Further, using this representation we show that the intrinsic risk measure yields a smaller amount needed to reach acceptability while ensuring equal performance.
In the setting of convex acceptance sets, \hyperref[Sec:DualRepres]{Section \ref*{Sec:DualRepres}} starts with a short summary of standard duality results of traditional risk measures, followed by the derivation of a dual representation for intrinsic risk measures.
Finally, concluding remarks and a short outlook regarding possible extensions and questions are given in \hyperref[Sec:outlook]{Section \ref*{Sec:outlook}}.
Throughout this article, we illustrate the new concepts and results using the Value at Risk acceptance set and demonstrate the calculation and application of intrinsic risk measures.

\section{Terminology and preliminaries}
\label{Section:Prelim}
In this section, we introduce common terminology, and the general notion of acceptance sets and traditional risk measures.
The aim is to establish a basis on which we can build our framework. 
At the end of this section, a motivational outlook for the intrinsic risk measure is provided.\\

\noindent Throughout this study we will work on an atomless probability space $(\Omega,\mathcal{F},\mathbb{P})$.
For the sake of simplicity and presentational flow we consider financial positions on the space of essentially bounded random variables $\mathcal{X} = L^\infty(\Omega,\mathcal{F},\mathbb{P})$ endowed with the $\mathbb{P}$-almost sure order and the $\mathbb{P}$-essential supremum norm.
However, the majority of the results can be stated for bounded random variables on a model-free measurable space $(\Omega,\mathcal{F})$, or even in greater generality on arbitrary ordered real topological vector spaces.
We will explicitly indicate where immediate extensions are possible.

\subsection{Acceptance sets}
\label{sec:accSets}

In the financial world, it is a central task to hold positions that satisfy certain acceptability criteria, may they represent own preferences or be of regulatory nature. 
These criteria can be brought into a mathematical framework via so-called acceptance sets. 
The following definition determines a very general structure of acceptance sets which reflects the `minimal' human rationale. 

\begin{definition}
A subset $\mathcal{A} \subset \mathcal{X}$ is called an \emph{acceptance set} if it satisfies
\begin{itemize}[label = \raisebox{0.13ex}{\footnotesize $\bullet$}]
\item \emph{Non-triviality}: $\acc \neq \emptyset$ and $\mathcal{A} \varsubsetneq \mathcal{X}$, and \label{Property:nonTrivial}
\item \emph{Monotonicity}: $X_T \in \mathcal{A}$, $Y_T \in \mathcal{X}$, and $Y_T \geq X_T$ imply $Y_T \in \mathcal{A}$. \label{Property:monotonicity}
\end{itemize}
An element $X_T \in \acc$ is called $\acc$\emph{-acceptable}, or just \emph{acceptable} if the reference to $\acc$ is clear.
Similarly, $X_T \notin \acc$ is said to be $(\acc\text{-})$\emph{unacceptable}.
\end{definition}
\noindent Non-triviality is mathematically important and also representative of real world requirements, as we will not just accept any bad situation and on the other hand, since any event requires near-term reactions, there must always be acceptable actions. 
Monotonicity implements the idea that any financial position dominating an acceptable position with respect to the order $\leq$ must also be acceptable. 

These two axioms constitute the basis for acceptance sets.
Depending on the context, it is often necessary to impose further structure and we recall three for our framework relevant properties. 
$\acc$ is called

\begin{itemize}[label = \raisebox{0.13ex}{\footnotesize $\bullet$}]
\item a \emph{cone} or \emph{conic} if $X_T \in \acc$ implies that for all $\lambda > 0$ also $\lambda X_T \in \acc$, 
\item \emph{convex} if $X_T, Y_T \in \acc$ implies that for all $\lambda \in [0,1]$ also $\lambda X_T + (1-\lambda)Y_T \in \acc$,
\item \emph{closed} if $\acc = \bar{\acc}$.
\end{itemize}
The cone property allows for arbitrary scaling of financial positions invariant of their acceptability status. 
Convexity represents the principle of diversification: given two acceptable positions, any convex combination of these will be acceptable. 
It will be discussed in \hyperref[Section:tradRM]{Section \ref*{Section:tradRM}} how these two properties translate to monetary risk measures. 
Closedness is of importance when considering limits of sequences of acceptable positions.
Apart from this, it is economically motivated as it prohibits arbitrarily small perturbations to benefit unacceptable positions and make them acceptable. 

\vspace{8pt}
The next lemma summarises some useful properties of acceptance sets which will be used in subsequent sections.

\begin{lemma} \label{Lemma:acceptanceSet1}
Let $\mathcal{A} \subset \mathcal{X}$ be an acceptance set. 
Then the following assertions hold.
\begin{enumerate}
\item \label{LemmaPart:allConst} $\mathcal{A}$ contains all sufficiently large constants and no sufficiently small constants. 
\item \label{LemmaPart:intAcc} $S_T \in \mathrm{int}(\mathcal{A})$ if and only if there exists an $\varepsilon > 0$ such that $S_T - \varepsilon \mathbf{1}_\Omega \in \mathcal{A}$. 
\item \label{LemmaPart:intcl} The interior $\mathrm{int}(\mathcal{A})$ and the closure $\bar{\mathcal{A}}$ are acceptance sets, and $\mathrm{int}(\acc) = \mathrm{int}(\bar{\acc})$. 
\item \label{LemmaPart:intclAlsoConic} If $\mathcal{A}$ is a cone, then $\mathrm{int} (\acc)$ and $\bar{\acc}$ are cones, and $0 \notin \mathrm{int}(\acc)$ and $0 \in \bar{\acc}$. 
\end{enumerate}
\end{lemma}
\begin{proof}
\begin{enumerate}[wide=0pt]
\item Since $\acc$ is a nonempty, proper subset of $\mathcal{X}$, monotonicity implies that any constant dominating some $X_T \in \acc$ is contained in $\acc$ and no constant dominated by some $Y_T \notin \acc$ is contained in $\acc$. 
\item Assume $S_T - \varepsilon \mathbf{1}_\Omega \in \mathcal{A}$, for some $\epsilon > 0$. 
Since for $\delta \in (0, \varepsilon)$ and $X_T$ in the ball $B_\delta(S_T) = \{ Y_T \in \mathcal{X} \,|\, \Vert Y_T - S_T \Vert_{L^\infty(\mathbb{P})} < \delta \}$ the inequality $X_T - (S_T -\varepsilon \mathbf{1}_\Omega) \geq (\varepsilon-\delta) \mathbf{1}_\Omega > 0$ holds, monotonicity implies $S_T \in \mathrm{int}(\mathcal{A})$.
The other direction follows directly from the definition of the interior.
\item The proof of \hyperref[LemmaPart:intcl]{Assertion~\ref*{LemmaPart:intcl}} follows the lines of the proofs of Lemma 2.3 (iv) and (v) in \cite{bib:WPC1} and is omitted.
\item Given $S_T \in \mathrm{int}(\acc)$, \hyperref[LemmaPart:intAcc]{Assertion~\ref*{LemmaPart:intAcc}} and the cone property yield $\lambda(S_T - \varepsilon \mathbf{1}_\Omega) \in \mathcal{A}$, for some $\varepsilon > 0$ and all $\lambda > 0$.
Then the other direction of \hyperref[LemmaPart:intAcc]{Assertion~\ref*{LemmaPart:intAcc}} implies $\lambda S_T \in \mathrm{int}(\acc)$.
Given $S_T \in \bar{\acc}$, take a sequence $\{S_T^n\}_{n\in\mathbb{N}} \subset \mathcal{A}$ with limit $S_T$.
Then conicity implies $\{\lambda S_T^n\}_{n\in\mathbb{N}} \subset \acc$, for any $\lambda > 0$, and we conclude that $\lambda S_T$ belongs to $\bar{\mathcal{A}}$.
To show that $0 \notin \mathrm{int}(\acc)$ assume the contrary.
Then by \hyperref[LemmaPart:intAcc]{Assertion~\ref*{LemmaPart:intAcc}}, we find an $\varepsilon > 0$ such that $-\varepsilon \mathbf{1}_\Omega \in \acc$.
The cone property and monotonicity imply $\acc = \mathcal{X}$, a contradiction.
For the last part take a decreasing sequence $\{\lambda_n\}_{n \in \mathbb{N}}$ converging to $0$ and $S_T \in \acc$.
Conicity implies $\{\lambda_n S_T\}_{n \in \mathbb{N}} \subset \acc$ and we conclude that $0 \in \bar{\acc}$.  \qedhere
\end{enumerate}
\end{proof}

\begin{remark}
\hyperref[Lemma:acceptanceSet1]{Lemma \ref*{Lemma:acceptanceSet1}} can be stated in a model-free environment of a measurable space $(\Omega,\mathcal{F})$, see \cite[Lemma 2.3]{bib:WPC1}.
Moreover, for the most part it can be extended to general ordered topological vector spaces.
See for example sections 2 and 3 in \cite{bib:WPC2}.
However, for general spaces it is necessary to substitute the interior by a refined concept such as the core.
\end{remark}

We conclude this subsection with the well-known example of the so-called Value at Risk acceptance set. 
\begin{example}[Value at Risk acceptance] \label{Ex:unmotivatedVaRAccSet}
For any probability level $\alpha \in \big(0,\frac{1}{2}\big)$ the set
\begin{align*}
\acc_\alpha = \{X_T \in \mathcal{X} \,|\,  \mathbb{P}[X_T < 0] \leq \alpha \}
\end{align*}
defines a closed, conic acceptance set which, in general, is not convex.
\end{example}
\begin{proof}
A few short calculations show that $\acc_\alpha$ is a conic acceptance set.
To show that $\acc_\alpha$ is closed in $L^\infty(\mathbb{P})$ consider a sequence $\{X_T^n\}_{n \in \mathbb{N}} \subset \acc_\alpha$ converging to some $X_T$.
For any $\delta > 0$ and any $n \in \mathbb{N}$ the following inequality holds,
\begin{align*}
\mathbb{P}[X_T < -\delta] &= \mathbb{P}[X_T < - \delta \,, X_T^n < -\tfrac{\delta}{2}] + \mathbb{P}[X_T < - \delta \,, X_T^n \geq -\tfrac{\delta}{2}] \\ &\leq \alpha + \mathbb{P}[ |X_T^n-X_T| > \tfrac{\delta}{2}] \,.
\end{align*}
Since norm convergence implies convergence in probability, we can let $n \rightarrow \infty$ and get $\mathbb{P}[X_T < -\delta] \leq \alpha$.
It follows $\mathbb{P}[X_T < 0] = \lim_{\delta \rightarrow 0} \mathbb{P}[X_T < -\delta] \leq \alpha$.
In order to show that $\acc_\alpha$ is not convex, we use conicity to reduce the problem to finding $X_T,Y_T \in \acc_\alpha$ such that $X_T+Y_T \notin \acc_\alpha$.
For two disjoint subsets $A,B \in \mathcal{F}$ with $\mathbb{P}[A] = \mathbb{P}[B] = \alpha$ the choices $X_T = -\mathbf{1}_{A}$ and $Y_T = -\mathbf{1}_{B}$ yield the desired inequality.
\end{proof}

\begin{remark}
This example can directly be extended to $L^p(\Omega,\mathcal{F},\mathbb{P})$, for $p \in [0,\infty)$.
Details can be found in the dissertation of C.~Munari \cite[Section 2.4.1]{bib:Munari}.
\end{remark}

\subsection{Traditional risk measures}  \label{Section:tradRM}
This subsection introduces the notion of traditional risk measures as commonly used by financial institutions.
Acceptance sets determine the meaning of `good' and `bad'. 
Traditional risk measures refine this differentiation and allow us to rank financial positions with respect to their distance in direction $r \mathbf{1}_\Omega$ (or $S_T/S_0$) to the acceptance set.
To clearly distinguish between these risk measures and intrinsic risk measures, we define the class of traditional risk measures following \cite[Definition 2.1]{bib:ADHE}.

\begin{definition}
A \emph{traditional risk measure} is a map from $\mathcal{X}$ into $\mathbb{R}$.
\end{definition} 
\noindent In \hyperref[Sec:IntrinsicRisk]{Section \ref*{Sec:IntrinsicRisk}}, we will see that intrinsic risk measures are defined on $\fspace$.\\

In what follows we recall some well-known traditional risk measures.
For this let $X_T, Y_T$ and $\mathbf{r} = r\mathbf{1}_{\Omega}$ be elements of $\mathcal{X}$, and let $\rho$ denote a traditional risk measure.

\subsubsection{Coherent risk measures}
Coherent risk measures form the historical foundation of the modern risk measure theory.
P.~Artzner, F.~Delbaen, J.~Eber, and D.~Heath \cite{bib:ADHE} define them by the following set of axioms. 
A traditional risk measure is called \emph{coherent} if it satisfies
\begin{itemize}[label = \raisebox{0.13ex}{\footnotesize $\bullet$}]
\item \emph{Monotonicity}: $X_T \geq Y_T$ implies $\rho(X_T) \leq \rho(Y_T)$,
\item \emph{Cash-additivity}: For $m \in \mathbb{R}$ we have $\rho(X_T + m \mathbf{r}) = \rho(X_T) - m$,
\item \emph{Positive Homogeneity}: For $\lambda \geq 0$ we have $\rho(\lambda X_T) = \lambda \rho(X_T)$, and
\item \emph{Subadditivity}: $\rho(X_T+Y_T) \leq \rho(X_T) + \rho(Y_T)$.
\end{itemize}
Decreasing monotonicity allows us to rank financial positions according to their risk.
It is cash-additivity that constitutes the basis for the interpretation of a risk measure as an additionally required amount of capital.
Adding this capital to the financial position, its risk becomes $0$, as by cash-additivity, $\rho(X_T + \rho(X_T) \mathbf{r}) = 0$. 
These assumptions seem natural in the context of capital requirements and they truly characterise the term \emph{monetary risk measures}, as coined by H.~F\"ollmer and A.~Schied in \cite[Definition 4.1]{bib:FS}. 

\subsubsection{Convexity of risk measures} 
Positive homogeneity, however, may not be satisfied, as risk can behave in a non-linear way.
A possible variation is the following property around which H.~F\"ollmer and A.~Schied \cite{bib:FS} base their discussion of risk measures.

\begin{itemize}[label = \raisebox{0.13ex}{\footnotesize $\bullet$}]
\item \emph{Convexity}: If $\lambda \in [0,1]$, then $\rho(\lambda X_T + (1-\lambda)Y_T) \leq \lambda \rho(X_T) + (1-\lambda)\rho(Y_T)$.
\end{itemize}
A short calculation reveals that under positive homogeneity, subadditivity and convexity are equivalent.
H. F\"ollmer and A. Schied \cite[Definition 4.4]{bib:FS} decide to drop the homogeneity axiom and replace subadditivity by convexity, and call the result a \emph{convex measure of risk} -- a measure which becomes coherent if the assumption of positive homogeneity is added.\\

Interestingly, the axioms we have seen so far form a canonical connection to our acceptance sets.

\begin{proposition} \label{Prop:CorrespondenceOfRiskAndAcc}
Any monetary risk measure $\rho: \mathcal{X} \rightarrow \mathbb{R}$ defines an acceptance set 
\begin{align} \label{eq:accViaRho}
\acc_\rho = \{X_T \in \mathcal{X} \,|\, \rho(X_T) \leq 0 \}\,.
\end{align}
Moreover, if $\rho$ is positive homogeneous, then $\acc_\rho$ is a cone, and if $\rho$ is convex, then $\acc_\rho$ is convex.

On the other hand, each acceptance set $\acc$ defines a monetary risk measure via 
\begin{align} \label{eq:rhoViaAcc}
\rho_{\acc}(X_T) = \inf\{m \in \mathbb{R} \,|\, X_T + m \mathbf{r} \in \acc\}\,.
\end{align}
Similarly, if $\acc$ is a cone, then $\rho_\acc$ is positive homogeneous, and if $\acc$ is convex, then $\rho_\acc$ is convex.

In particular, this means $\rho_{\acc_\rho} = \rho$ and $\acc \subseteq \acc_{\rho_\acc}$, with equality $\acc= \acc_{\rho_\acc}$ if the acceptance set is closed.
\end{proposition}
\begin{proof}
The proof is analogous to the proofs of Proposition 4.6 and Proposition 4.7 in \cite{bib:FS} for bounded measurable functions on $(\Omega,\mathcal{F})$.
\end{proof}

\hyperref[Prop:CorrespondenceOfRiskAndAcc]{Proposition \ref*{Prop:CorrespondenceOfRiskAndAcc}} allows us to define acceptance sets via known risk measures and vice versa, as is illustrated in the following example. 
\hyperref[Prop:CorrespondenceOfRiskAndAcc]{Proposition \ref*{Prop:CorrespondenceOfRiskAndAcc}} can be stated for more general spaces $\mathcal{X}$ and eligible assets, for this see \hyperref[Prop:CorrespondenceOfRiskAndAcc2]{Proposition \ref*{Prop:CorrespondenceOfRiskAndAcc2}}.

\begin{example}[Value at Risk acceptance] \label{Def:VaRandES}
For a given probability level $\alpha \in \big(0, \frac{1}{2}\big)$ we define the risk measure \emph{Value at Risk} for all random variables on $(\Omega,\mathcal{F})$ by
\begin{align*}
\text{VaR}_\alpha(X_T) = \inf \{ m \in \mathbb{R} \,|\, \mathbb{P}[ X_T + m < 0] \leq \alpha \}\,,
\end{align*}
the negative of the upper $\alpha$-quantile of $X_T$.
Corresponding to \hyperref[Prop:CorrespondenceOfRiskAndAcc]{Proposition \ref*{Prop:CorrespondenceOfRiskAndAcc}}, the $\text{VaR}_\alpha$-acceptance set is given by
\begin{align*} 
\acc_\alpha := \acc_{\mathrm{VaR}_\alpha} = \{X_T \in \mathcal{X} \,|\, \text{VaR}_\alpha(X_T) \leq 0 \} \,.
\end{align*}
Let us recall the closed, conic set $\{X_T \in \mathcal{X} \,|\,  \mathbb{P}[X_T < 0] \leq \alpha \}$ from \hyperref[Ex:unmotivatedVaRAccSet]{Example \ref*{Ex:unmotivatedVaRAccSet}}. The risk measure defined by this set via \hyperref[eq:rhoViaAcc]{Equation (\ref*{eq:rhoViaAcc})} is just the Value at Risk.
So we conclude that $\acc_\alpha = \{X_T \in \mathcal{X} \,|\,  \mathbb{P}[X_T < 0] \leq \alpha \}$ and that $\mathrm{VaR}_\alpha$ is a positive homogeneous monetary risk measure which, in general, is not convex, and thus, not coherent.
\end{example}

\subsubsection{Cash-subadditivity and quasi-convexity of risk measures} 
N.~El Karoui and C.~Ravanelli \cite{bib:EKR} point out that in presence of stochastic interest rates the axiom of cash-additivity relies on the assumption that the discounting process does not carry additional risk, since a financial position is discounted prior to applying the risk measure.
To relax this restriction they introduce the property of cash-subadditivity, where the equality in the cash-additivity condition is changed to the inequality `$\geq$'.
However, S.~Cerreia-Vioglio, F.~Maccheroni, M.~Marinacci and L.~Montrucchio \cite{bib:CMMM} explain that under cash-subadditivity, convexity is not a rigorous representative of the diversification principle which translates into the following requirement for risk measures.
\begin{itemize}[label = \raisebox{0.13ex}{\footnotesize $\bullet$}]
\item \emph{Diversification Principle}: If $\rho(X_T),\rho(Y_T) \leq \rho(Z_T)$ is satisfied, then for all $\lambda \in [0,1]$ also $\rho(\lambda X_T + (1-\lambda)Y_T) \leq \rho(Z_T)$ holds.
\end{itemize}
Substituting $\rho(Z_T)$ by $\max\{\rho(X_T), \rho(Y_T)\}$ yields the equivalent and recently importance gaining property of  
\begin{itemize}[label = \raisebox{0.13ex}{\footnotesize $\bullet$}]
\item \emph{Quasi-convexity}: If $\lambda \in [0,1]$, then $\rho(\lambda X_T + (1-\lambda)Y_T) \leq \max\{\rho(X_T), \rho(Y_T)\}$.
\end{itemize}
Interestingly, quasi-convexity is equivalent to convexity under cash-additivity, since for any two positions with $\rho(X_T) \leq \rho(Y_T)$ we find an $m \in \mathbb{R}_{\geq 0}$ such that $\rho(X_T - m \mathbf{r}) = \rho(Y_T)$ so that we get for any $\lambda \in (0,1)$
\begin{align*}
\rho(\lambda X_T + (1-\lambda) Y_T) + \lambda m &\leq \max\{ \rho(X_T - m\mathbf{r}),\rho(Y_T)\} \\
 &= \lambda \rho(X_T) + (1-\lambda) \rho(Y_T) + \lambda m \,.
\end{align*}
This equivalence does not hold under cash-subadditivity as shown with all details in \cite[Example 2.10]{bib:AS}, resulting in the necessity to explicitly implement the diversification principle and thus, in the introduction of cash-subadditive, quasi-convex risk measures.

\subsubsection{General monetary risk measures} 
However, stochastic interest rates can be directly addressed with risk measures of the form introduced and treated in \cite{bib:WPC2} and \cite{bib:WPC1},  
\begin{align} \label{Eq:GenMonRisk}
\rho_{\acc,S}(X_T) = \inf \Big\{m \in \mathbb{R} \,\big|\, X_T + \frac{m}{S_0}S_T \in \acc \Big\}\,.
\end{align}
This approach avoids discounting, since the stochastic eligible asset is incorporated into the risk measure. 
Moreover, C.~Munari provides a broad discussion of the discounting argument, revealing further fundamental issues with discounting on the level of acceptance sets in Section 1.3 of \cite{bib:Munari}.

\hyperref[Eq:GenMonRisk]{Equation (\ref*{Eq:GenMonRisk})} defines a generalised monetary risk measures which satisfies the following requirement for a specific eligible asset $S = (S_0, S_T )$,
\begin{itemize}[label = \raisebox{0.13ex}{\footnotesize $\bullet$}]
\item \emph{S-additivity}: If $m \in \mathbb{R}$, then $\rho(X_T + mS_T ) = \rho(X_T) - m S_0$.
\end{itemize}
Also this general setup yields the equivalence of quasi-convexity and convexity, and it exhibits a similar correspondence between acceptance sets and risk measures. 
The following result extends \hyperref[Prop:CorrespondenceOfRiskAndAcc]{Proposition \ref*{Prop:CorrespondenceOfRiskAndAcc}} to stochastic eligible assets.
It will be used in \hyperref[Sec:IntrinsicRiskonPosHom&Cashadditive]{Section \ref*{Sec:IntrinsicRiskonPosHom&Cashadditive}} and \hyperref[Sec:DualRepres]{Section \ref*{Sec:DualRepres}} to relate intrinsic to traditional risk measures. 

\begin{proposition} \label{Prop:CorrespondenceOfRiskAndAcc2}
\hyperref[Prop:CorrespondenceOfRiskAndAcc]{Proposition \emph{\ref*{Prop:CorrespondenceOfRiskAndAcc}}} holds true if we replace $L^\infty(\Omega,\mathcal{F},\mathbb{P})$ by any real ordered topological vector space, cash-additivity by $S$-additivity, and \hyperref[eq:rhoViaAcc]{Equation \emph{(\ref*{eq:rhoViaAcc})}} by \hyperref[Eq:GenMonRisk]{Equation \emph{(\ref*{Eq:GenMonRisk})}}, for any eligible asset $S = (S_0,S_T) \in \mathbb{R}_{>0} \times \acc$.
\end{proposition}
\begin{proof}
See the proofs of propositions 3.2.3, 3.2.4, 3.2.5, and 3.2.8 in \cite{bib:Munari}.
The second claim in \hyperref[Prop:CorrespondenceOfRiskAndAcc]{Proposition \emph{\ref*{Prop:CorrespondenceOfRiskAndAcc}}} follows from two short calculations. 
\end{proof}

\subsection{Back to the financial drawing board}
All risk measures in the previous section have their foundation in the idea to add additional capital through eligible assets or directly as money to the existing financial position.
Consequently, the procedure to make an unacceptable position $X_T$ \mbox{acceptable} is to raise at least the `minimal' required capital $\rho_{\acc,S}(X_T)$ and invest it into $S$.
But the problems of providing capital and the risk of failing to obtain enough are not addressed in this approach.
In the literature, authors even concentrate on ensuring that this amount stays finite. 
Apart from the construal that the financial position cannot reach the acceptance set, infinite values have no practical effect for applications and the operational interpretation as additional capital gets lost.

In the next section, we propose a different action. 
As mentioned in the introduction, since the traditional risk measures are defined via hypothetical additional capital which is not always available in reality, we introduce in what follows a risk measure that only allows the usage of internal capital contained in the financial position.
We propose to use a prespecified amount of available capital, the current value of the financial position, and invest it into an eligible asset.
This approach has two simultaneous effects on the altered position.
Firstly, selling portions of a position, we reduce the potential risk therein. 
Secondly, we invest the acquired capital into an eligible asset which by definition has acceptable risk.
This procedure yields a convex combination of the financial position and a multiple of the eligible asset.
The result is a more direct path towards the acceptance set and a less costly action, its overall  cost being bounded by the initial value of the financial position.

\section{The intrinsic risk measure} \label{Sec:IntrinsicRisk} 
In light of the motivation given above, we introduce the new approach in this section.
\mbox{Following} the framework of general eligible assets we define the intrinsic risk measure in \hyperref[subsec:IntroductionIntrinsicRisk]{Section \ref*{subsec:IntroductionIntrinsicRisk}}. Intuition behind the approach and differences to traditional risk measures are illustrated in \hyperref[fig:visualExample2]{Figure \ref*{fig:visualExample2}}.
In \hyperref[subsec:ClassificationIntrinsicRisk]{Section \ref*{subsec:ClassificationIntrinsicRisk}}, emerging properties of this measure are studied in juxtaposition with the ones of traditional measures. 

\subsection{Introducing the intrinsic risk measure} \label{subsec:IntroductionIntrinsicRisk}
From here on, we explicitly consider a one period economy with inception at time $0$ and maturity at time $T$.
Eligible assets as used in \hyperref[Eq:GenMonRisk]{Equation (\ref*{Eq:GenMonRisk})} are given by pairs $S = (S_0,S_T)$.
This characteristic is now extended to any financial position.

\begin{definition}
Financial positions are defined on the product space $\fspace$.
\begin{enumerate}
\item Call $X = (X_0,X_T) \in \mathbb{R}_{>0} \times \mathcal{X}$ an \emph{extended financial position}.
The number $X_0$ denotes the price of the position at inception and $X_T$ denotes the random payoff or net worth of the position at maturity.
\item Given an acceptance set $\acc$, call $S = (S_0,S_T) \in \mathbb{R}_{>0} \times \acc$ an \emph{extended eligible asset} if $S_T \geq 0$.
\end{enumerate}
We shall keep the terms \emph{financial position} and \emph{eligible asset} to denote their extended versions.
Random variables always have a subscript $T$, so that there is no risk of confusion.
\end{definition}
\vspace{8pt}

On this basis, we can now introduce the announced class.

\begin{definition}[Intrinsic Risk Measure] \label{Def:IntrinsicRiskMeasure}
Let $\mathcal{A} \subset \mathcal{X}$ be an acceptance set and let $S \in \mathbb{R}_{>0} \times \mathcal{A}$ be an eligible asset.  \\
An \emph{intrinsic risk measure} is a map $R_{\mathcal{A},S} : \mathbb{R}_{>0} \times \mathcal{X} \rightarrow [0,1]$ defined by
\begin{align} \label{Eq:IntrinsicRisk}
\riskR{A}{S}{X} = \riskset{X}{S} \,.
\end{align}
\end{definition}

\begin{remark}
The functional introduced above measures risk as the smallest percentage $\riskR{A}{S}{X}$ of a financial position $X$ that has to be sold at inception for the amount $ X_0 \riskR{A}{S}{X}$, and in turn has to be reinvested into the eligible asset $S$ so that the resulting position becomes acceptable. 
\end{remark}

We provide a first intuition for how these measures operate and what the main differences compared to traditional risk measures are in the following example.
\begin{example} \label{Ex:visualExample}
Consider \hyperref[fig:visualExample2]{Figure \ref*{fig:visualExample2}} below and let $\acc$ be an arbitrary closed acceptance set.
The traditional approach illustrated in \hyperref[fig:visualExample2Sub1]{Figure \ref*{fig:visualExample2}(a)} yields an acceptable altered position \mbox{$X_T^\rho :=  X_T + \frac{\rho_{\acc,S}(X_T)}{S_0} S_T$}. The intrinsic risk measure illustrated in \hyperref[fig:visualExample2Sub2]{Figure \ref*{fig:visualExample2}(b)} gives us \mbox{$X_T^R := (1-\riskR{A}{S}{X}) X_T + \riskR{A}{S}{X} \frac{X_0}{S_0} S_T$}.
We point out two characteristics of intrinsic risk measures:
\begin{itemize}[label = \raisebox{0.13ex}{\footnotesize $\bullet$}]
\item The position $X_T^R$ lies on the boundary of $\acc$, just as $X_T^\rho \in \partial\acc$ in the traditional approach.
This will be proven in \hyperref[Prop:ResultingPosition]{Proposition \ref*{Prop:ResultingPosition}}.
\item If additionally $\acc$ is a cone as in \hyperref[fig:visualExample2]{Figure \ref*{fig:visualExample2}}, we see that $X_T^\rho$ is a multiple of $X_T^R$.
Indeed, if $R_{\acc,S}(X) \in (0,1)$, then we have the relation
\begin{align*}
X_T^R = (1-R_{\acc,S}(X)) X_T^\rho \,.
\end{align*}
For more details see \hyperref[Cor:rhoInTermsOfR]{Corollary \ref*{Cor:rhoInTermsOfR}} and \hyperref[Rem:X^RScaledX^rho]{Remark \ref*{Rem:X^RScaledX^rho}}.
\end{itemize} 
\begin{figure}[h] 
\centering
\begin{subfigure}{.5\textwidth}
  \centering
  \includegraphics[scale= 0.32]{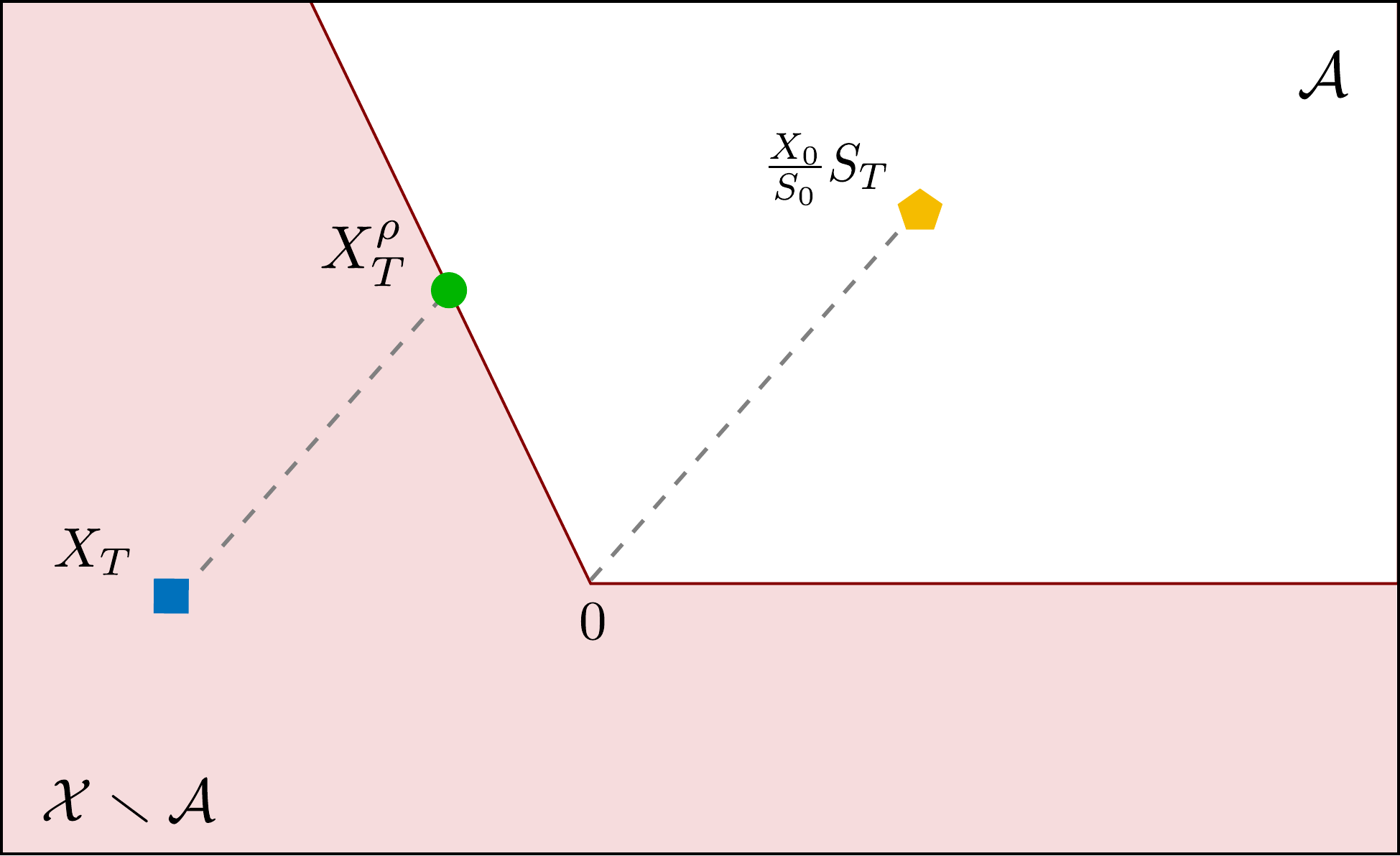}
  \caption{Traditional approach}
  \label{fig:visualExample2Sub1}
\end{subfigure}%
\begin{subfigure}{.5\textwidth}
  \centering
  \includegraphics[scale= 0.32]{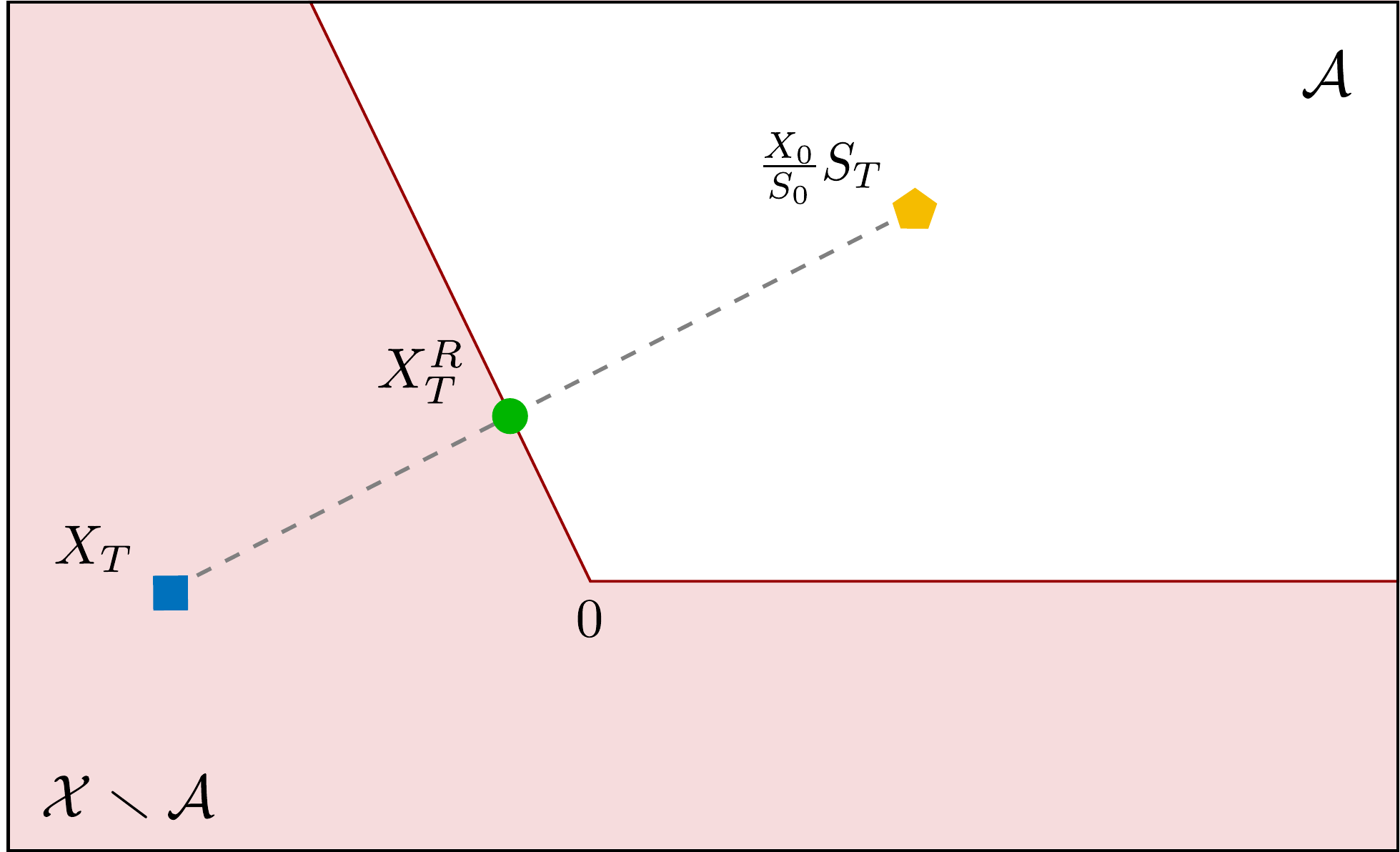}
  \caption{Intrinsic approach}
  \label{fig:visualExample2Sub2}
\end{subfigure}
\caption{The payoff of the eligible asset (yellow  \pentagon) is used to make the unacceptable position (blue \Square) acceptable (green \Circle).}
\label{fig:visualExample2}
\end{figure}
\end{example}

The following proposition provides two conditions for well-definedness.

\begin{proposition} \label{Prop:Existenz}
Let $\acc \subset \mathcal{X}$ be an acceptance set and $S$ an eligible asset.
If $\mathcal{A}$ is a cone or $\acc$ contains $0$, then $R_{\acc,S}$ is well-defined on $\mathbb{R}_{>0} \times \mathcal{X}$.
\end{proposition}
\begin{proof}
As we search for the smallest $\lambda \in [0, 1]$ such that $(1-\lambda)X_T + \lambda \frac{X_0}{S_0} S_T$ belongs to $\mathcal{A}$, investing all of $X_0$ into the eligible asset, that means choosing $\lambda = 1$, will always result in an acceptable position as long as $\frac{X_0}{S_0} S_T \in \acc$.
\begin{enumerate}[wide=0pt]
\item If $\acc$ is a cone and $S_T \in \acc$, then $\lambda S_T \in \acc$, for any $\lambda > 0$.
\item Since $0 \in \acc$ and $S_T \geq 0$, monotonicity of $\acc$ implies that $\lambda S_T \in \mathcal{A}$, for any $\lambda \geq 0$.
\end{enumerate}
Since by definition $X_0, S_0 > 0$, the assertions follow. 
\end{proof}

\begin{remark} \label{Rem:InitialRemark}
\begin{asparaenum}
\item The assumption $0 \in \acc$ already found its way into the financial literature.
For example, it is equivalent to Axiom 2.1 in \cite{bib:ADHE} that all non-negative random variables are contained in $\acc$.
Furthermore, the \emph{normalisation property} $\rho(0) = 0$ introduced in \cite[below Definition 4.1]{bib:FS} implies $0 \in \acc$ if the acceptance set is closed.
In this case, also conicity of $\acc$ implies $0 \in \acc$. 
\item Usually, at least one of the conditions in \hyperref[Prop:Existenz]{Propositions \ref*{Prop:Existenz}} is required.
More restrictive acceptability criteria which regulate size and reject positive positions must be handled with care.
In this case, it can be of interest to restrict the definition of $R_{\acc,S}$ to $\mathbb{R}_{>c} \times \mathcal{X}$ for some constant $c > 0$ such that $\frac{c}{S_0} S_T \in \acc$. 
\item In contrast to traditional risk measures, intrinsic risk measures cannot attain infinite values, making the question of finiteness superfluous.
If $\acc$ is a cone, \hyperref[Prop:eligibleAssetInfluence]{Proposition \ref*{Prop:eligibleAssetInfluence}} or alternatively \hyperref[Prop:IntrinsicRepresPosHom&S-Additive]{Theorem \ref*{Prop:IntrinsicRepresPosHom&S-Additive}} reveal a connection between traditional risk measures taking the value $+\infty$ and intrinsic risk measures being equal to $1$. 
However, infinite values have little meaning in applications, whereas an intrinsic risk of $1$ suggests a practicable transition from $X_T$ to $\frac{X_0}{S_0} S_T$. 
\item Another salient difference to traditional risk measures is that intrinsic risk measures do not take negative values. 
While traditional risk measures exhibit a form of symmetry around the boundary of the acceptance set, translating financial positions towards that boundary, irrespective of their acceptability, intrinsic risk measures do not alter acceptable positions.
In this way the principal objective in working with acceptability criteria is emphasised, only moving unacceptable positions into the acceptance set.
\item The product $X_0 \riskR{A}{S}{X} \in [0,X_0]$ of the intrinsic risk with the initial value $X_0$ is a monetary amount.
This will be used to compare intrinsic risk measures to monetary risk measures in nominal terms in \hyperref[Cor:MonetaryComparison]{Corollary \ref*{Cor:MonetaryComparison}}.
\item As a direct consequence of the definition, we have for $\alpha X = (\alpha X_0, \alpha X_T)$ with $\alpha > 0$, that $\riskR{A}{S}{\alpha X} = R_{\alpha^{-1}\mathcal{A},S}(X)$.
This equality allows us to define the intrinsic risk measure in terms of returns of financial positions, setting $\alpha = (X_0)^{-1}$.
In particular, if $\acc$ is a cone, then the measure effectively operates on returns, as it is \emph{scale-invariant}. 
We will address this property again in \hyperref[Cor:scaleInv]{Corollary \ref*{Cor:scaleInv}}.
Nevertheless, we consider the first approach to be more transparent and we will explicitly keep track of the initial values $X_0$ and $S_0$.
\end{asparaenum}
\end{remark}

\begin{example} \label{Ex:generalExample}
Consider the Value at Risk acceptance set $\acc_\alpha = \{X_T \in \mathcal{X} \,|\,  \mathbb{P}[X_T < 0] \leq \alpha \}$ from \hyperref[Def:VaRandES]{Example \ref*{Def:VaRandES}}.
Denote by $F_X$ the continuous cumulative distribution function of $X_T$ with inverse $F_X^{-1}$.
Assume $X_T \notin \acc_\alpha$, this means $F_X^{-1}(\alpha) < 0$, and let $S_T  = r S_0 \mathbf{1}_{\Omega} > 0$ be constant, this implies $S_T \in \mathrm{int}(\acc_\alpha)$ (see \cite[Proposition 3.6]{bib:WPC2} or for $L^p$ spaces see \cite[Lemma 4.1]{bib:WPC1}).
We will see in \hyperref[subsec:ClassificationIntrinsicRisk]{Section \ref*{subsec:ClassificationIntrinsicRisk}} that in this case, we can restrict $\lambda$ to $(0,1)$, so that we can write 
\begin{align*}
R_{\acc_\alpha,S}(X) &= \inf \big\{ \lambda \in [0,1] \,|\, X_T^{\lambda,S} \in \acc_\alpha \big\} = \inf \big\{ \lambda \in (0,1) \,|\, \mathbb{P}[X^{\lambda,S}_T < 0] \leq \alpha \big\} \\
&= \inf \big\{ \lambda \in (0,1) \,|\, F_X(-(1-\lambda)^{-1}\lambda r X_0) \leq \alpha \big\} = \frac{F_X^{-1}(\alpha)}{F_X^{-1}(\alpha) -  r X_0} \\
&= \frac{\text{VaR}_\alpha(X_T)}{r X_0 + \text{VaR}_\alpha(X_T)} \,.
\end{align*}
So when we use $\acc_\alpha$ to define the intrinsic risk measure, then we can write it as a function of the corresponding traditional risk measure.
A more direct and general derivation will be presented in \hyperref[Prop:IntrinsicRepresPosHom&S-Additive]{Theorem \ref*{Prop:IntrinsicRepresPosHom&S-Additive}}. 
\end{example}

\subsection{Classification of intrinsic risk measures} \label{subsec:ClassificationIntrinsicRisk}
In this section, we aim to compare intrinsic risk measures to traditional ones by means of their properties.
We begin with geometric properties and their connections to the image of the intrinsic risk measure.
Afterwards, we consider monotonicity, a translation relation, and we conclude this section with a discussion of quasi-convexity.

\vspace{8pt}
For the subsequent study it is convenient to define intermediate positions.
\begin{definition}
For $\alpha \in [0, 1]$ define the \emph{intermediate position (between $X$ and $S$)}
\begin{align*}
X^{\alpha,S} = (X_0, X_T^{\alpha,S}) = (X_0, (1-\alpha) X_T + \alpha \tfrac{X_0}{S_0} S_T ) \in \mathbb{R}_{>0} \times \mathcal{X}\,,
\end{align*}
a position originating from $X$ which has been shifted towards the payoff of the eligible asset $S$. 
The line segment $\{X_T^{\alpha,S}\,|\, \alpha \in [0,1]\}$ is illustrated by the dashed line in \hyperref[fig:visualExample2Sub2]{Figure \ref*{fig:visualExample2}(b)}.
Of particular interest are shifted positions $X^{\riskR{A}{S}{X},S}$ which we abbreviate by $X^{R(X),S}$, whenever the reference to $\acc$ and $S$ is clear.
\end{definition}

\subsubsection{Connections between intrinsic risk measures and acceptance sets}

To find a basis for comparison, we examine the relationship between intrinsic risk measures and their defining acceptance sets. 

\begin{proposition}[Relevance] \label{Prop:GeneralRelevance}
Let $\acc$ be a closed acceptance set.
For any financial position $X \in \mathbb{R}_{>0} \times \mathcal{X}$ it holds that $\riskR{A}{S}{X} > 0$ if and only if $X_T \notin \mathcal{A}$.
\end{proposition}
\begin{proof}
Since $\acc$ is closed, $\riskR{A}{S}{X} = 0$ implies $X^{0,S}_T = X_T \in \acc$.
Conversely, if $X_T \in \acc$, then $\inf\{ \lambda \in [0,1] \,|\, (1-\lambda)X_T + \lambda \frac{X_0}{S_0} S_T \in \acc \} = 0$. 
\end{proof}

The next proposition proves the first assertion in \hyperref[Ex:visualExample]{Example \ref*{Ex:visualExample}}.

\begin{proposition} \label{Prop:ResultingPosition}
Let $\acc$ be a closed acceptance set.
For any financial position $X$ with $\riskR{A}{S}{X} \in (0,1)$ we have $X^{R(X),S}_T \in \partial \acc$.
\end{proposition}
\begin{proof}
Since $\acc$ is closed, $X_T^{R(X),S}$ belongs to $\acc$ by definition.
Now contradictorily assume that $X_T^{R(X),S}$ lies in the interior of $\acc$.
Then there exists a $\delta > 0$ such that the ball $B_\delta(X_T^{R(X),S}) \subset \mathrm{int}(\acc)$.
In particular, for $0<\varepsilon < \delta \Vert X_T - \frac{X_0}{S_0} S_T \Vert^{-1}_\infty$ the element $X_T^{R(X)-\varepsilon,S}$ lies in $B_\delta(X_T^{R(X),S})$, hence, $R_{\acc,S}(X)$ is not the infimum in contradiction to its definition. 
\end{proof}
\noindent Note that for $X$ with $\riskR{A}{S}{X} \in \{0, 1\}$ the intermediate position does not a priori belong to the boundary of $\acc$, since $X^{0,S}_T = X_T$ and $X^{1,S}_T = \frac{X_0}{S_0} S_T$ could belong to $\mathcal{A}$. 
However, we will see in \hyperref[Prop:eligibleAssetInfluence]{Proposition \ref*{Prop:eligibleAssetInfluence}} that if $\acc$ is a cone, the case $\riskR{A}{S}{X} = 1$ can only occur if $S_T \in \partial \acc$.\\
 
In the next result, we show that we can shift $X_T^{R(X),S}$ from the boundary further towards $\frac{X_0}{S_0} S_T$ without leaving a convex or conic acceptance set.
\begin{proposition} \label{Prop:weiterRein}
Let the closed acceptance set $\mathcal{A}$ be either a cone, or convex with $0 \in \acc$ and let $S$ be an eligible asset.
Then $\{ X^{\alpha,S}_T \,|\, \alpha \in [R_{\acc,S}(X),1]\} \subset \acc$.
\end{proposition}
\begin{proof}
From \hyperref[Prop:ResultingPosition]{Proposition \ref*{Prop:ResultingPosition}} we know that $X_T^{R(X),S} \in \acc$, not restricting $\riskR{A}{S}{X}$ to $(0,1)$.
If $\acc$ is convex, then for $\beta \in [0, 1]$ the element $(1 - \beta) X^{R(X),S}_T + \beta \frac{X_0}{S_0} S_T$ lies in $\acc$. 
But this is just a reparametrisation of $X_T^{\alpha,S}$ under the map $[0, 1] \rightarrow [\riskR{A}{S}{X}, 1]$ defined by $\beta \mapsto (1 - \riskR{A}{S}{X}) \beta + \riskR{A}{S}{X} =: \alpha$, so the assertion follows. \\
If $\mathcal{A}$ is a cone and not convex, then $\lambda X^{R(X),S}_T \in \acc$, for all $\lambda \geq 0$.
For any $\alpha \in [\riskR{A}{S}{X}, 1]$ we set $\lambda = \frac{1-\alpha}{1-\riskR{A}{S}{X}}$ and observe $\alpha - \lambda \riskR{A}{S}{X} = \frac{\alpha - \riskR{A}{S}{X}}{1-\riskR{A}{S}{X}} \in [0,1]$.
Since $S_T \geq 0$ we get the inequality 
\begin{align*}
X^{\alpha,S}_T = \lambda X^{R(X), S}_T + (\alpha - \lambda \riskR{A}{S}{X}) \frac{X_0}{S_0} S_T \geq \lambda X^{R(X),S}_T \,, 
\end{align*}
and monotonicity yields the assertion. 
\end{proof}

In \hyperref[Prop:ResultingPosition]{Proposition \ref*{Prop:ResultingPosition}}, we assumed that $\riskR{A}{S}{X} < 1$. 
As mentioned above, this inequality depends on the choice of the eligible asset. 
\hyperref[Prop:eligibleAssetInfluence]{Proposition \ref*{Prop:eligibleAssetInfluence}} is a direct analogy with certain finiteness results for traditional risk measures, see \hyperref[Rem:AnalogyWithFinitenessResults]{Remark \ref*{Rem:AnalogyWithFinitenessResults}}.

\begin{proposition} \label{Prop:eligibleAssetInfluence}
Let $\mathcal{A}$ be a closed, conic acceptance set, and let $S$ be an eligible asset. 
Then $R_{\mathcal{A},S} < 1$ on $\mathbb{R}_{>0} \times \mathcal{X} \setminus{\acc}$ if and only if $S_T \in \mathrm{int}(\mathcal{A})$.
\end{proposition}
\begin{proof}
Assume $\riskR{A}{S}{X} < 1$, for all $X \in \mathbb{R}_{>0} \times \mathcal{X} \setminus{\acc}$.
As a consequence of \hyperref[LemmaPart:intclAlsoConic]{Lemma \ref*{Lemma:acceptanceSet1} Assertion \ref*{LemmaPart:intclAlsoConic}}, no strictly negative constant is contained in $\acc$.
In particular, the element \mbox{$C = (C_0,C_T) := (c,-c\mathbf{1}_\Omega)$} lies in $\mathbb{R}_{>0} \times \mathcal{X} \setminus \acc$, for some constant $c > 0$. 
\hyperref[Prop:ResultingPosition]{Proposition \ref*{Prop:ResultingPosition}} implies now
\begin{align*}
C_T^{R(C),S} = -c   (1-\riskR{A}{S}{C}) \mathbf{1}_\Omega +  \riskR{A}{S}{C} \frac{c}{S_0} S_T \in \acc\,.
\end{align*}
Using the cone property we know that
\begin{align*}
\frac{S_0}{c \riskR{A}{S}{C}} C_T^{R(C),S} = - \frac{S_0 (1 -\riskR{A}{S}{C})}{\riskR{A}{S}{C}}  \mathbf{1}_\Omega +  S_T \in \acc \,.
\end{align*}
But $S_0 > 0$ and $\riskR{A}{S}{C} \in (0,1)$, so \hyperref[LemmaPart:intAcc]{Lemma \ref*{Lemma:acceptanceSet1} Assertion \ref*{LemmaPart:intAcc}} implies $S_T \in \mathrm{int}(\mathcal{A})$. \\
For the `if' direction assume $S_T \in \mathrm{int}(\mathcal{A})$.
By \hyperref[LemmaPart:intcl]{Lemma \ref*{Lemma:acceptanceSet1} Assertion \ref*{LemmaPart:intcl}}, also $X^{1,S}_T \in \mathrm{int}(\mathcal{A})$.
Similar to the proof of \hyperref[Prop:ResultingPosition]{Proposition \ref*{Prop:ResultingPosition}} we find $\delta > 0$ and $0 < \varepsilon < \delta \Vert X_T - \frac{X_0}{S_0} S_T \Vert_{\infty}^{-1}$ such that $X^{1-\varepsilon , S}_T \in B_\delta(\frac{X_0}{S_0} S_T) \subset \mathrm{int}(\acc)$.
We conclude that $\riskR{A}{S}{X} < 1$. 
\end{proof}

\begin{remark}\label{Rem:AnalogyWithFinitenessResults}
\hyperref[Prop:eligibleAssetInfluence]{Proposition \ref*{Prop:eligibleAssetInfluence}} is an analogue of certain finiteness results for traditional risk measures, as the following equivalence holds,
\begin{align*}
\rho_{\acc,S}(X_T) \text{ is finite if and only if } S_T \in \mathrm{int}(\acc)\,.
\end{align*}
For details see Theorem 3.3 and Corollary 3.4 in \cite{bib:WPC1}, or in greater generality, for arbitrary (pre-)ordered topological vector spaces, a similar result is true as shown in Proposition 3.1 in \cite{bib:WPC2}.
The explicit connection of intrinsic risk measures equal to $1$ and traditional risk measure taking the value $+ \infty$ will be established in \hyperref[Prop:IntrinsicRepresPosHom&S-Additive]{Theorem \ref*{Prop:IntrinsicRepresPosHom&S-Additive}} by an representation result on conic acceptance sets, also providing an alternative proof of \hyperref[Prop:eligibleAssetInfluence]{Proposition \ref*{Prop:eligibleAssetInfluence}} via traditional risk measures.
\end{remark}

\subsubsection{Monotonicity of intrinsic risk measures}
In this section, we show that monotonicity of $\acc$ implies that intrinsic risk measures are decreasing functionals, similarly to the traditional approach.
However, on the product space $\mathbb{R}_{>0} \times \mathcal{X}$ we can consider two compatible orderings,
\begin{itemize}[label = \raisebox{0.13ex}{\footnotesize $\bullet$}]
\item \emph{Element-wise domination:} $X \geqslant_{\text{el}} Y$ if $X_0 \geq Y_0$ and $X_T \geq Y_T$, and
\item \emph{Return-wise domination:} $X \geqslant_{\text{ret}} Y$ if $\frac{X_T}{X_0} \geq  \frac{Y_T}{Y_0}$.
\end{itemize}
The first ordering describes the situation in which position $X$ dominates $Y$ at both inception and maturity.
However, from financial point of view it can be restrictive to penalise an institution whose position has low initial value, but outperforms at maturity.
The second ordering incorporates this case and orders positions with regard to their returns.

\begin{proposition}[Monotonicity] \label{Prop:Monotonicity}
Let $\acc$ be an acceptance set containing $0$, let $S \in \aspace$ be an eligible asset and let $X, Y \in \fspace$.
\begin{enumerate}
\item If $X \geqslant_{\text{el}} Y$, then $\riskR{A}{S}{X} \leq \riskR{A}{S}{Y}$.
\item If $\acc$ is additionally a cone, then $X \geqslant_{\text{ret}} Y$ implies $\riskR{A}{S}{X} \leq \riskR{A}{S}{Y}$.
\end{enumerate}
\end{proposition}
\begin{proof}
Take any $\lambda \in [0, 1]$ such that $Y^{\lambda,S}_T \in \acc$.
\begin{enumerate}[wide=0pt]
\item Since $S_T \geq 0$, element-wise domination implies $X_T^{\lambda,S} \geq Y_T^{\lambda,S}$, and thus by monotonicity of $\acc$, $X^{\lambda,S}_T \in \acc$ which means $\riskR{A}{S}{X} \leq \riskR{A}{S}{Y}$.
\item Assume $X$ dominates $Y$ return-wise. 
Conicity implies $\frac{X_0}{Y_0} Y^{\lambda,S}_T \in \acc$. 
Now
\begin{align*}
\frac{X_0}{Y_0} Y^{\lambda,S}_T = X_0 \Big((1 - \lambda) \frac{Y_T}{Y_0} + \lambda \frac{S_T}{S_0} \Big) \leq X_0 \Big((1 - \lambda)\frac{X_T}{X_0} + \lambda \frac{S_T}{S_0} \Big) = X^{\lambda,S}_T \,,
\end{align*} 
so monotonicity of $\acc$ implies $X^{\lambda,S}_T \in \acc$, and thus, $\riskR{A}{S}{X} \leq \riskR{A}{S}{Y}$.  \qedhere
\end{enumerate}
\end{proof}

\subsubsection{A transition property and quasi-convexity of intrinsic risk measures}

We start this section with a transition property that can be linked to $S$-additivity of traditional risk measures.
However, since intrinsic risk measures operate in terms of percentages and not monetary amounts, a simple addition of units of $S_T$ is not appropriate.
Rather, we compare the risks of an intermediate position $X^{\alpha,S}$ and its original position $X$.
In the second part, we show that intrinsic risk measures are quasi-convex.

\begin{proposition} \label{Prop:IntrinsicInvariance}
For $X \in \mathbb{R}_{>0} \times {\mathcal{X}\setminus \acc}$ and $\alpha \in [0, \riskR{A}{S}{X}]$ we have
\begin{align} \label{Eq:TranslationProp}
R_{\acc,S}(X^{\alpha,S}) = \frac{\riskR{A}{S}{X} - \alpha}{1-\alpha} \,.
\end{align}
\end{proposition}
\begin{proof}
Fix an $\alpha \in [0, \riskR{A}{S}{X}]$.
Using the bijection $[0,1] \rightarrow [\alpha,1]$ given by $\lambda \mapsto (1-\lambda)\alpha + \lambda =: \bar{\lambda}$, a direct calculation yields
\begin{align*}
R_{\acc,S}(X^{\alpha,S}) &= \inf \left\{ \lambda \in [0,1] \,\Big|\, (1-\lambda) X^{\alpha,S}_T + \lambda \frac{X_0}{S_0} S_T \in \acc \right\} \\
&= \inf \left\{ \lambda \in [0,1] \,\Big|\, (1-\lambda)(1-\alpha) X_T + \big((1-\lambda)\alpha + \lambda \big) \frac{X_0}{S_0} S_T \in \acc \right\} \\
&= \frac{1}{1-\alpha} \left( \inf \left\{ \bar{\lambda} \in [\alpha,1] \,\Big|\, (1-\bar{\lambda}) X_T + \bar{\lambda} \frac{X_0}{S_0} S_T \in \acc \right\} - \alpha \right) \\
&= \frac{\riskR{A}{S}{X} - \alpha}{1-\alpha} \,.
\end{align*}
The last equality holds, since $\alpha \leq \riskR{A}{S}{X}$. 
\end{proof}

Next, we discuss convexity and quasi-convexity.
Convex combinations on the product space $\fspace$ are understood element-wise and we write
\begin{align*}
\alpha X + (1-\alpha) Y := ( \alpha X_0 + (1-\alpha) Y_0 \, , \, \alpha X_T + (1-\alpha) Y_T) \in \fspace\,.
\end{align*}
So considering the convex combination of two future positions $X_T$ and $Y_T$ results in taking the convex combination of their initial values $X_0$ and $Y_0$.\\

In contrast to traditional risk measures, convexity of the acceptance set implies quasi-convexity of the intrinsic risk measures. 

\begin{proposition}[Quasi-convexity] \label{Prop:QuasiConvexity}
Let $\acc$ be a closed, convex acceptance set containing $0$, and let $S$ be an eligible asset. 
Then $R_{\acc,S}$ is quasi-convex, that means for all $\alpha \in [0,1]$, and any $X,Y \in \fspace$
\begin{align*}
\riskR{A}{S}{\alpha X + (1-\alpha) Y} \leq \max\{\riskR{A}{S}{X}, \riskR{A}{S}{Y}\} \,,
\end{align*} 
and in general, $R_{\acc,S}$ is not convex.
\end{proposition}
\begin{proof}
\begin{enumerate}[wide=0pt]
\item First we show that $R_{\acc,S}$ is not convex.
To this end consider a position $X$ with $\riskR{A}{S}{X} > 0$ and an $\alpha \in (0,\riskR{A}{S}{X})$.
A short calculation yields that the convex combination $(1-\beta)X + \beta X^{\alpha,S}$ is equal to $X^{\alpha \beta,S}$.
We can now use \hyperref[Eq:TranslationProp]{Equation (\ref*{Eq:TranslationProp})} to show that $(1-\beta) \riskR{A}{S}{X} + \beta \riskR{A}{S}{X^{\alpha,S}} \ngeq \riskR{A}{S}{X^{\alpha \beta,S}}$.
For example, take an $X$ with $\riskR{A}{S}{X} = \frac{3}{4}$, and let $\alpha = \beta = \frac{1}{2}$.
Then we have $\riskR{A}{S}{X^{\alpha,S}} = \frac{\riskR{A}{S}{X}-\alpha}{1-\alpha} = \frac{1}{2}$ and $\riskR{A}{S}{X^{\alpha \beta, S}} = \frac{2}{3}$. 
Hence, 
\begin{align*}
\riskR{A}{S}{(1-\beta)X + \beta X^{\alpha,S}} = \frac{2}{3} > \frac{5}{8} = (1-\beta) \riskR{A}{S}{X} + \beta \riskR{A}{S}{X^{\alpha,S}} \,,
\end{align*} 
showing that $R_{\acc,S}$ is not convex.
\item Now we show quasi-convexity.
Assume without loss of generality $\riskR{A}{S}{X} \leq \riskR{A}{S}{Y}$, or else exchange $X$ and $Y$. 
By \hyperref[Prop:ResultingPosition]{Proposition \ref*{Prop:ResultingPosition}} and since we allow $\riskR{A}{S}{Y} \in \{0,1\}$, we have $Y_T^{R(Y),S} \in \acc$.
Since $\riskR{A}{S}{X} \leq \riskR{A}{S}{Y}$, \hyperref[Prop:weiterRein]{Proposition \ref*{Prop:weiterRein}} yields $X_T^{R(Y),S} \in \acc$, and since $\acc$ is convex, also $\alpha X_T^{R(Y),S} + (1-\alpha) Y_T^{R(Y),S} \in \acc$, for all $\alpha \in [0, 1]$. 
So for any $\alpha \in [0, 1]$, the random variable
\begin{align} \label{Eq:ConvexCombination}
(1-\lambda)(\alpha X_T + (1-\alpha) Y_T) + \lambda \frac{\alpha X_0 + (1-\alpha) Y_0}{S_0} S_T = \alpha X^{\lambda,S}_T + (1-\alpha) Y^{\lambda,S}_T
\end{align}
belongs to $\acc$ whenever $\lambda \in [\riskR{A}{S}{Y},1]$.
Hence, for all $\alpha \in [0,1]$ we have
\begin{align*}
\riskR{A}{S}{\alpha X + (1-\alpha)Y} \stackrel{(\ref{Eq:ConvexCombination})}{=} &\inf \left\{ \lambda \in [0,1]\,|\, \alpha X^{\lambda,S}_T + (1-\alpha) Y^{\lambda,S}_T \in \acc \right\} \\
\leq \;\,\, &\riskR{A}{S}{Y} = \max \{\riskR{A}{S}{X}, \riskR{A}{S}{Y} \} \,,
\end{align*}
showing quasi-convexity of the intrinsic risk measure.  \qedhere
\end{enumerate} 
\end{proof}

We show that convexity of $\acc$ is necessary for quasi-convexity of the intrinsic risk measure, using the non-convex Value at Risk acceptance set from \hyperref[Ex:unmotivatedVaRAccSet]{Example \ref*{Ex:unmotivatedVaRAccSet}}.
\begin{example}
Consider the closed, conic acceptance set $\acc_\alpha$ from \hyperref[Ex:unmotivatedVaRAccSet]{Example \ref*{Ex:unmotivatedVaRAccSet}}.
We showed that $X_T = -\mathbf{1}_A$, $Y_T = -\mathbf{1}_B$ for disjoint $A, B \in \mathcal{F}$ with $\mathbb{P}[A]=\mathbb{P}[B]= \alpha$ are contained in $\acc_\alpha$, whereas $\frac{1}{2}(X_T + Y_T)$ is not.
With \hyperref[Prop:GeneralRelevance]{Proposition \ref*{Prop:GeneralRelevance}} we conclude that $\riskAR{\acc_\alpha}{S}{X} = \riskAR{\acc_\alpha}{S}{Y} = 0$, but $\riskAR{\acc_\alpha}{S}{\frac{1}{2}(X+Y)} > 0$, for any $S, X_0, Y_0$.
In fact, we will see in \hyperref[Sec:IntrinsicRiskonPosHom&Cashadditive]{Section \ref*{Sec:IntrinsicRiskonPosHom&Cashadditive}}, that if for example $\frac{S_T}{S_0} = \mathbf{1}_\Omega$, the intrinsic risk of $\frac{1}{2}(X + Y)$ is $\frac{1}{X_0+Y_0+1}$.
Indeed, setting $\lambda = \frac{1}{X_0 + Y_0 + 1}$ yields the probability $\mathbb{P}[X_T + Y_T < - 1]$, where $-1$ is in fact the largest value $x$ such that $\mathbb{P}[X_T + Y_T < x] \leq \alpha$.
\end{example}

Quasi-convexity also holds true for convex combinations of eligible assets. 

\begin{proposition} \label{Prop:QuasiconvexityEligibleAsset}
Let $\acc$ be a closed, convex acceptance set containing $0$, and fix a financial position $X \in \fspace$.
For any $\alpha \in [0,1]$ and any two eligible assets $S^1,S^2 \in \aspace$ with same initial price $P = S^1_0 = S^2_0$ we have,
\begin{align*}
\riskR{A}{\alpha S^1 + (1-\alpha)S^2}{X} \leq \max\{\riskR{A}{S^1}{X}, \riskR{A}{S^2}{X}\}\,.
\end{align*}
\end{proposition}
\begin{proof}
Similar to the proof of \hyperref[Prop:QuasiConvexity]{Proposition \ref*{Prop:QuasiConvexity}}, assume that $\riskR{A}{S^1}{X} \leq \riskR{A}{S^2}{X}$, so that $X_T^{\riskR{A}{S^2}{X},S^2}, \, X_T^{\riskR{A}{S^2}{X},S^1} \in \acc$.
In particular, any convex combination of the two belongs to $\acc$.
Observing that for any $\lambda \in [0,1]$
\begin{align*}
\alpha X^{\lambda,S^1}_T + (1-\alpha) X^{\lambda,S^2}_T = (1-\lambda) X_T + \lambda \frac{X_0}{P} (\alpha S^1_T + (1-\alpha) S^2_T ) = X^{\lambda, \alpha S^1 + (1-\alpha)S^2}_T \,,
\end{align*}
the assertion follows as in the proof of \hyperref[Prop:QuasiConvexity]{Proposition \ref*{Prop:QuasiConvexity}}. 
\end{proof}

\begin{remark}
These results describe a symmetry of the risk measure, in the sense that convex combinations can be taken both inside and outside the acceptance set without losing control over the risk.
\hyperref[Prop:QuasiConvexity]{Proposition \ref*{Prop:QuasiConvexity}} and \hyperref[Prop:QuasiconvexityEligibleAsset]{Proposition \ref*{Prop:QuasiconvexityEligibleAsset}} together imply that the risk of any convex combination of $X,Y$ using any convex combination of eligible assets $S^1,S^2$ is bounded by the maximum of the four values $\riskR{A}{S^1}{X}$, $\riskR{A}{S^2}{X}$, $\riskR{A}{S^1}{Y}$, and $\riskR{A}{S^2}{Y}$. 
\end{remark}

\section{Representation and efficiency on conic acceptance sets} \label{Sec:IntrinsicRiskonPosHom&Cashadditive}
The intrinsic risk measure derived in \hyperref[Ex:generalExample]{Example \ref{Ex:generalExample}} can be written as a function of its traditional counterpart.
We will see in this section that if the underlying acceptance set is a closed cone, \hyperref[Prop:CorrespondenceOfRiskAndAcc]{Proposition \ref*{Prop:CorrespondenceOfRiskAndAcc}} and \hyperref[Prop:CorrespondenceOfRiskAndAcc2]{Proposition \ref*{Prop:CorrespondenceOfRiskAndAcc2}} yield an alternative representation.
Using this representation we show that the intrinsic approach requires overall less capital compared to the traditional approach and at the same time it yields positions with same performance.

\subsection{Alternative representation of intrinsic risk measures}

We will now see how the correspondence between traditional risk measures and their underlying acceptance sets allows for an alternative representation of intrinsic risk measures. 

\begin{theorem}[Representation on cones] \label{Prop:IntrinsicRepresPosHom&S-Additive}
Let $\acc$ be a closed, conic acceptance set and let $\rho_{\acc,S}(X_T) = \inf \{ m \in \mathbb{R} \,|\, X_T + \frac{m}{S_0} S_T \in \acc \}$. 
The intrinsic risk measure with respect to the same acceptance set and eligible asset can be written as
\begin{align} \label{Eq:IntrinsicRepresPosHom&S-Additive}
\riskAR{\acc}{S}{X} = \frac{(\rho_{\acc,S}(X_T))^+}{X_0 + \rho_{\acc,S}(X_T)} \,.
\end{align}
\end{theorem}
\begin{proof}
Consider \hyperref[Prop:CorrespondenceOfRiskAndAcc2]{Proposition \ref{Prop:CorrespondenceOfRiskAndAcc2}}.
Since $\acc$ is closed, we have $\acc = \acc_{\rho_{\acc,S}}$, implying
\begin{align*}
\riskAR{\acc}{S}{X} &= \inf\{\lambda \in [0,1] \,|\, X_T^{\lambda,S} \in \acc \} =  \inf\{\lambda \in [0,1] \,|\, \rho_{\acc,S}(X_T^{\lambda,S}) \leq 0 \} \,.
\end{align*}
Now $\rho_{\acc,S}$ is $S$-additive and positive homogeneous. 
Using this we can rearrange to get
\begin{align*}
\riskAR{\acc}{S}{X} = \inf \big\{ \lambda \in [0,1] \,|\, \rho_{\acc,S}(X_T) \leq \lambda \big(X_0 + \rho_{\acc,S}(X_T)\big) \big\} \,.
\end{align*}
Observe that $\rho_{\acc,S}(X_T) \leq 0$ implies $\riskAR{\acc}{S}{X} = 0$.
If on the other hand $\riskA{\acc}{S}{X_T} > 0$, then we can solve for $\lambda$ to get the form in \hyperref[Eq:IntrinsicRepresPosHom&S-Additive]{Equation (\ref*{Eq:IntrinsicRepresPosHom&S-Additive})}. 
\end{proof}

\begin{remark} \label{Rem:HebbareSingularity}
Note that the singularity at $\rho_{\acc,S}(X_T) = - X_0$ in \hyperref[Eq:IntrinsicRepresPosHom&S-Additive]{Equation (\ref*{Eq:IntrinsicRepresPosHom&S-Additive})} is removable. We can either explicitly define $R_{\acc,S}(X) = 0$ whenever $X_T \in \acc$ or write 
\begin{align*}
\riskAR{\acc}{S}{X} = \frac{(\rho_{\acc,S}(X_T))^+}{X_0 + (\rho_{\acc,S}(X_T))^+} \,.
\end{align*}
For this reason we keep the more concise notation in \hyperref[Eq:IntrinsicRepresPosHom&S-Additive]{Equation (\ref*{Eq:IntrinsicRepresPosHom&S-Additive})}.
\end{remark}

\begin{corollary} \label{Cor:RepresentingIntrinsicRiskviaPosHomandMonInv}
\begin{enumerate}[wide=0pt]
\item Considering $\rho_{\acc,\bar{r}}(X_T) = \inf \{ m \in \mathbb{R}\,|\, X_T + m \bar{r} \mathbf{1}_\Omega \in \acc\}$ and $S=(S_0,r S_0)$, the intrinsic risk measure takes the form 
\begin{align} \label{intrinsicRiskOtherRepresentation}
\riskAR{\acc}{S}{X} = \frac{(\rho_{\acc,\bar{r}}(X_T))^+}{\frac{r}{\bar{r}} X_0 + \rho_{\acc,\bar{r}}(X_T)} \,.
\end{align}
\item Value at Risk is defined by $\rho_{\acc_\alpha}(X) = \inf \{ m \in  \mathbb{R} \,|\, X_T + m \mathbf{1}_\Omega \in \acc_\alpha \}$, with $\bar{r} = 1$ in \hyperref[intrinsicRiskOtherRepresentation]{Equation (\ref*{intrinsicRiskOtherRepresentation})}, so that it takes the form derived in \hyperref[Ex:generalExample]{Example \emph{\ref*{Ex:generalExample}}}.
\end{enumerate}
\end{corollary}

\vspace{6pt}
In our opinion, \hyperref[Prop:IntrinsicRepresPosHom&S-Additive]{Theorem \ref*{Prop:IntrinsicRepresPosHom&S-Additive}} plays a central role, since it allows us to draw direct connections to traditional risk measures and use results developed in this field.
In particular, the subsequent results hold for the most frequently used acceptance sets in practice, namely the Value at Risk and Expected Shortfall acceptance sets, since they are cones.\\
As an example we show when the intrinsic risk measure defined on cones is a continuous functional and less then 1.
\begin{corollary}\label{Cor:Consequences}
Let $\acc$ be a closed, conic acceptance set.  
\begin{enumerate}
\item $R_{\acc,S} < 1$ on $\fspace \setminus \acc$ if and only if $S_T \in \mathrm{int}(\acc)$.
\item If $S_T \in \mathrm{int}(\mathcal{X}_+)$, then $R_{\acc,S}$ is continuous on $\fspace$.
\item If $\acc$ is additionally convex, then $S_T \in \mathrm{int}(\acc)$ implies continuity of $R_{\acc,S}$.
\end{enumerate}
\end{corollary}
\begin{proof}
\begin{enumerate}[wide=0pt]
\item We have already seen this result in \hyperref[Prop:eligibleAssetInfluence]{Proposition \ref*{Prop:eligibleAssetInfluence}}.
\hyperref[Prop:IntrinsicRepresPosHom&S-Additive]{Theorem \ref*{Prop:IntrinsicRepresPosHom&S-Additive}} can be used in combination with Theorem 3.3 in \cite{bib:WPC1} for an alternative proof.
\item \label{proof:continuity} The map $f: (x_0,x) \mapsto \frac{x^+}{x_0 + x}$ is jointly continuous on $\mathbb{R}_{>0} \times \mathbb{R}$, see also \hyperref[Rem:HebbareSingularity]{Remark \ref*{Rem:HebbareSingularity}}.
By Proposition 3.1 in \cite{bib:WPC2}, $\rho_{\acc,S}$ is (Lipschitz-) continuous on $\mathcal{X}$ if $S_T \in \mathrm{int}(\mathcal{X}_+)$.
As the composition of two continuous maps $X \mapsto (X_0, \rho_{\acc,S}(X_T)) \mapsto f(X_0,\rho_{\acc,S}(X_T))$ the intrinsic risk measures is continuous on $\fspace$.
\item Theorem 3.16 in \cite{bib:WPC2} gives us continuity of $\rho_{\acc,S}$ on $\mathcal{X}$. 
The assertion follows as in part \ref{proof:continuity}.  \qedhere
\end{enumerate} 
\end{proof}

\begin{remark}
Note that if state the first assertion in \hyperref[Cor:Consequences]{Corollary \ref*{Cor:Consequences}} not as an equivalence, we can use the whole space $\fspace$ as: $S_T \in \mathrm{int}(\acc)$, then $R_{\acc,S} < 1$ on $\fspace$. 
The applied results from \cite{bib:WPC2} are true for general ordered topological vector spaces, and thus also \hyperref[Cor:Consequences]{Corollary \ref*{Cor:Consequences}} can be extended in this respect.
\end{remark}

As mentioned in \hyperref[Rem:InitialRemark]{Remark \ref*{Rem:InitialRemark}}, \hyperref[Prop:IntrinsicRepresPosHom&S-Additive]{Theorem \ref*{Prop:IntrinsicRepresPosHom&S-Additive}} directly yields that intrinsic risk measures defined by conic acceptance sets are scale-invariant.
\begin{corollary} \label{Cor:scaleInv}
Let $\acc$ be a closed conic acceptance set. 
Then $R_{\acc,S}$ is scale-invariant.
\end{corollary}
\begin{proof}
Since $\acc$ is a cone, $\rho_{\acc,S}$ is positive homogeneous.
If $X_T \notin \acc$, then $\rho_{\acc,S}(X_T) > 0$, and for any $\alpha > 0$ we get with \hyperref[Prop:IntrinsicRepresPosHom&S-Additive]{Theorem \ref*{Prop:IntrinsicRepresPosHom&S-Additive}} 
\begin{align*}
\riskAR{\acc}{S}{\alpha X} = \frac{\rho_{\acc,S}(\alpha X_T)}{\alpha X_0 + \rho_{\acc,S}(\alpha X_T)} = \frac{\alpha \rho_{\acc,S}(X_T)}{\alpha(X_0 + \rho_{\acc,S}(X_T))} = \riskAR{\acc}{S}{X} \,.
\end{align*}
If $X_T \in \acc$, then $\alpha X_T \in \acc$, resulting in $\riskAR{\acc}{S}{X} = \riskAR{\acc}{S}{\alpha X} = 0$. 
\end{proof}

We can turn things around and represent monetary risk measures in terms of intrinsic risk measures.
\begin{corollary} \label{Cor:rhoInTermsOfR}
Let $\acc$ be a closed and conic acceptance set and let $S$ be an eligible asset such that $S_T \in \mathrm{int}(\acc)$.
On $\mathcal{X}\setminus \acc$, the traditional risk measure $\rho_{\acc,S}$ can be written in terms of the intrinsic risk measure $R_{\acc,S}$ as 
\begin{align} \label{Eq:rhoInTermsOfR}
\rho_{\acc,S}(X_T) = \frac{X_0 R_{\acc,S}(X)}{1-R_{\acc,S}(X)} \,,
\end{align}
for $X = (X_0,X_T) \in \mathbb{R}_{>0} \times \mathcal{X} \setminus \acc$.
\end{corollary}
\begin{proof}
For any $X_T \in \mathcal{X} \setminus \acc$ we have $\risk{A}{S}{X_T} > 0$ and by \hyperref[Prop:eligibleAssetInfluence]{Proposition \ref*{Prop:eligibleAssetInfluence}}, $S_T \in \mathrm{int}(\acc)$ means $R_{\acc,S} < 1$ on $\mathbb{R}_{>0} \times \mathcal{X} \setminus \acc$.
Setting $X = (X_0,X_T)$, for any $X_0 > 0$, and rearranging \hyperref[Eq:IntrinsicRepresPosHom&S-Additive]{Equation (\ref*{Eq:IntrinsicRepresPosHom&S-Additive})} yields the assertion. 
\end{proof}

\hyperref[Cor:rhoInTermsOfR]{Corollary \ref*{Cor:rhoInTermsOfR}} allows us to prove our claim in \hyperref[Ex:visualExample]{Example \ref*{Ex:visualExample}} that $X_T^\rho := X_T + \frac{\rho_{\acc,S}(X_T)}{S_0} S_T$ is a multiple of $X_T^{R(X),S}$.
\begin{corollary} \label{Rem:X^RScaledX^rho}
In the setting of \hyperref[Cor:rhoInTermsOfR]{Corollary \ref*{Cor:rhoInTermsOfR}}, consider \hyperref[Ex:visualExample]{Example \ref*{Ex:visualExample}}.
We have 
\begin{align*}
X_T^{R(X),S} = (1-R_{\acc,S}(X)) X_T^\rho \,.
\end{align*} 
\end{corollary}
\begin{proof}
Dividing $X_T^{R(X),S}$ by $1-R_{\acc,S}(X)$ and using \hyperref[Eq:rhoInTermsOfR]{Equation (\ref*{Eq:rhoInTermsOfR})} we get
\begin{align*}
\frac{1}{1-R_{\acc,S}(X)} X_T^{R(X),S} = X_T + \frac{X_0 R_{\acc,S}(X)}{(1-R_{\acc,S}(X))S_0} S_T = X_T + \frac{\rho_{\acc,S}(X_T)}{S_0} S_T = X_T^\rho \,,
\end{align*} 
the desired relation. 
\end{proof}

We raise the question whether the representation in \eqref{Eq:IntrinsicRepresPosHom&S-Additive} holds for convex acceptance sets which are not cones.
The next proposition shows that due to the loss of positive homogeneity the expression is only an upper bound of the intrinsic risk measure. 

\begin{proposition} \label{Prop:IntrinsicInequWhenConvex}
Let $\acc$ be a closed, convex acceptance set which is not a cone. 
Assume that $0 \in \acc$.
Then the following inequality holds,
\begin{align} \label{Eq:OnlyUpperBound}
\riskAR{\acc}{S}{X} \leq \frac{(\rho_{\acc,S}(X_T))^+}{X_0 + \rho_{\acc,S}(X_T)} \,,
\end{align}
with equality if $X_T \in \acc$.
\end{proposition}
\begin{proof}
By \hyperref[Prop:CorrespondenceOfRiskAndAcc2]{Proposition \ref*{Prop:CorrespondenceOfRiskAndAcc2}}, we know that $\rho_{\acc,S}$ is convex and $\acc = \acc_{\rho_{\acc,S}}$.
Convexity, $S$-additivity and the fact that $\rho_{\acc,S}(0) \leq 0$ imply the inequality
\begin{align*}
\rho_{\acc,S}(X_T^{\lambda,S}) \leq (1-\lambda) \rho_{\acc,S}(X_T) + \lambda \rho_{\acc,S}\big(\tfrac{X_0}{S_0} S_T\big) \leq (1-\lambda) \rho_{\acc,S}(X_T) - \lambda X_0 \,.
\end{align*}
From this we establish the following inclusion,
\begin{align*}
\{\lambda \in [0,1] \,|\, (1-\lambda)\rho_{\acc,S}(X_T) - \lambda X_0 \leq 0 \} \subseteq \{\lambda \in [0,1] \,|\, \rho_{\acc,S}(X_T^{\lambda,S}) \leq 0 \}\,,
\end{align*}
implying \eqref{Eq:OnlyUpperBound}. 
\end{proof}

\subsection{Efficiency of the intrinsic approach} \label{subsec:Efficiency}

The alternative representation in terms of traditional risk measures allows us to compare the costs of management actions of the intrinsic and the traditional approach.
For conic or convex acceptance sets the intrinsic risk measure suggests an action that requires overall less capital to be invested into the eligible asset, since part of the original position is sold.

\begin{corollary} \label{Cor:MonetaryComparison}
Let $\acc$ be a conic or convex acceptance set and closed in either case.  
Let $X$ be an unacceptable financial position and let $S$ be an eligible asset.
Then the inequality $X_0 \riskR{A}{S}{X} \leq \risk{A}{S}{X_T}$ holds.
\end{corollary}
\begin{proof}
If $\acc$ is convex, then \hyperref[Prop:IntrinsicInequWhenConvex]{Proposition \ref*{Prop:IntrinsicInequWhenConvex}} implies $X_0 \riskR{A}{S}{X} \leq X_0 \frac{\risk{A}{S}{X_T}}{X_0 + \risk{A}{S}{X_T}}$.
If $\acc$ is a cone, we have an equality by \hyperref[Prop:IntrinsicRepresPosHom&S-Additive]{Theorem \ref*{Prop:IntrinsicRepresPosHom&S-Additive}}.
The inequality $X_0 \frac{\risk{A}{S}{X_T}}{X_0 + \risk{A}{S}{X_T}} \leq \risk{A}{S}{X_T}$ always holds, proving the assertion. 
\end{proof}

\begin{remark}
We see that the required monetary amount to shift the financial position into the acceptance set is controlled by the initial value $X_0$.
But independent of the magnitude of $X_0$, the amount $X_0 \riskR{A}{S}{X}$, which is received for part of the initial position and invested into $S$, is always less than the amount $\risk{A}{S}{X_T}$, which is necessary if nothing of the initial position is sold.
\end{remark}

Since part of the position is sold, so in particular, possible profits are reduced, and less money is invested into the eligible asset, it is important to ascertain that the intrinsic risk measure yields altered positions that are not worse than those of the traditional approach. 
In the following, we discuss this matter in terms of returns.
For this let $X = (X_0,X_T)$ be a given financial position. 
Consulting the traditional risk measure we get the altered position $X_T^\rho := X_T + \frac{\risk{A}{S}{X_T}}{S_0} S_T \in \acc$.
So at inception, the total value of this position is \mbox{$X_0^\rho := X_0 + \risk{A}{S}{X_T}$}.
The return is then given by the fraction $X_T^\rho / X_0^\rho$.
If on the other hand we consult the intrinsic risk measure, the initial value $X_0$ is not changed and the return of the altered position is $X_T^{R(X),S} / X_0$.
The following table summarises these relations.

\begin{center}
\begin{tabular}{l | l l}
					& Intrinsic approach & Traditional approach \\ \hline
initial value		& $X_0$ 	& $X_0 + \rho_{\acc,S}(X_T)$ \\
altered position 	& $(1-R_{\acc,S}(X))X_T + R_{\acc,S}(X) \frac{X_0}{S_0} S_T$  & $X_T + \frac{\risk{A}{S}{X_T}}{S_0} S_T$ \\
return				& $X_T^{R(X),S} / X_0$	&	$X_T^\rho / X_0^\rho $\\
\end{tabular}
\end{center}
Interestingly, these two returns are equal if $\acc$ is conic.
\begin{corollary} \label{Prop:ReturnsEqual}
Let $\acc$ be a closed cone.  
Let $X$ be an unacceptable financial position and let $S$ be an eligible asset.
Then the returns of the positions $(X_0, X_T^{R(X),S})$ and $(X_0^\rho, X_T^{\rho})$ are equal.
\end{corollary}
\begin{proof}
By \hyperref[Rem:X^RScaledX^rho]{Corollary \ref*{Rem:X^RScaledX^rho}}, we have $X_T^{R(X),S} = (1-\riskR{A}{S}{X}) X_T^\rho$, and by \hyperref[Prop:IntrinsicRepresPosHom&S-Additive]{Theorem \ref*{Prop:IntrinsicRepresPosHom&S-Additive}}, we know that $1-\riskR{A}{S}{X} = \frac{X_0}{X_0 + \risk{A}{S}{X_T}}$.
Dividing both sides of the first equation by $X_0$ and using the second equality yield the assertion. 
\end{proof}

A popular way to examine the performance of an investment is to consider the so-called \emph{(revised) Sharpe ratio}.
\begin{example}
The revised Sharpe ratio is defined as the fraction of the expectation and the standard deviation of the excess return of an investment $X$ over a benchmark asset $B$,
\begin{align*}
\text{SR}_B(X) = \frac{\mathbb{E}[\frac{X_T}{X_0} - \frac{B_T}{B_0}]}{\sqrt{\mathbb{V}[\frac{X_T}{X_0} - \frac{B_T}{B_0}]}} \,.
\end{align*}
Since, by \hyperref[Prop:ReturnsEqual]{Corollary \ref*{Prop:ReturnsEqual}}, the returns of $(X_0, X_T^{R(X),S})$ and $(X_0^\rho, X_T^{\rho})$ are equal, also their Sharpe ratios are equal.
\end{example}

\section{Dual representations on convex acceptance sets} \label{Sec:DualRepres}
In this section, we derive a dual representation of intrinsic risk measures.
Firstly, we recall duality results of convex risk measures as in \cite[Section 4.3]{bib:FS} and derive a representation of coherent risk measures as in \cite[Section 4.1]{bib:ADHE} or its extended version due to F.~Delbaen \cite[Section 3]{bib:Delbaen2002} as a special case.
In the second part of this section, we derive an alternative representation of the acceptance set $\acc$, similar to that of S.~Drapeau and M.~Kupper \cite[Section 2]{bib:Drapeau}, which leads to the representation of intrinsic risk measures. 

\subsection{Duality of convex and coherent risk measures}
The standard approach for the derivation of a dual representation starts in a model-free environment of a measurable space $(\Omega,\mathcal{F})$ via finitely additive probability measures.
Under suitable continuity conditions the results can then be restricted to $\sigma$-additive probability measures.
The dual representation can then be translated to a probability space $(\Omega,\mathcal{F},\mathbb{P})$, see Section 4.2 and Section 4.3 in \cite{bib:FS}.
We skip these steps and directly consider $\mathcal{M}_{\sigma}(\mathbb{P}):= \mathcal{M}_{\sigma}(\Omega,\mathcal{F},\mathbb{P})$, the set of all $\sigma$-additive probability measures on $\mathcal{F}$ which are absolutely continuous with respect to $\mathbb{P}$ and financial positions in $\mathcal{X} = L^\infty(\Omega,\mathcal{F},\mathbb{P})$.\\

We start by recalling Theorem 4.31 in \cite{bib:FS}.
\begin{theorem} \label{Thm:ConvexEquivalencesOnProbSpace}
Let $\acc$ be a convex, $\text{weak}^*$-closed acceptance set.
Let $\rho_\acc$ be defined as in \hyperref[eq:rhoViaAcc]{Equation (\ref*{eq:rhoViaAcc})} with $\mathbf{r}=\mathbf{1}_\Omega$.
The risk measure has the representation 
\begin{align} \label{Eq:ConvexEquivalencesOnProbSpace}
\rho_\acc(X_T) = \sup_{\mathbb{Q}\in \mathcal{M}_{\sigma}(\mathbb{P})} (\mathbb{E}_{\mathbb{Q}}[-X_T] - \alpha_\mathrm{min}(\mathbb{Q},\acc))\,,
\end{align}
with the so-called \emph{minimal penalty function} $\alpha_{\mathrm{min}}$ defined for all $\mathbb{Q} \in \mathcal{M}_{\sigma}(\mathbb{P})$ by
\begin{align*}
\alpha_{\mathrm{min}}(\mathbb{Q},\acc) = \sup_{X_T \in \acc} \mathbb{E}_\mathbb{Q}[-X_T]\,.
\end{align*}
\end{theorem}

Since coherent risk measures are convex risk measures which are positively homogeneous, we immediately get the following corollary as a special case of \hyperref[Thm:ConvexEquivalencesOnProbSpace]{Theorem \ref*{Thm:ConvexEquivalencesOnProbSpace}}.
For further details see Corollary 4.18 and Corollary 4.34 in \cite{bib:FS}.
\begin{corollary} \label{Thm:DualRepresOfCoherentRisk}
Let $\acc$ be a conic, convex, $\text{weak}^*$-closed acceptance set.
Then, restricting probability measures to the subset $\mathcal{M} = \{\mathbb{Q} \in \mathcal{M}_{\sigma}(\mathbb{P}) \,|\, \alpha_{\mathrm{min}}(\mathbb{Q},\acc) = 0  \}$, the coherent risk measure $\rho_\acc:\mathcal{X} \rightarrow \mathbb{R}$ can be written as 
\begin{align*} 
\rho_\acc(X_T) = \sup_{\mathbb{Q} \in \mathcal{M}} \mathbb{E}_{\mathbb{Q}}[-X_T] \,. 
\end{align*}
\end{corollary}

\subsection{Duality of intrinsic risk measures}

In this section, we will start in a more general setting to illustrate the approach to this duality.
Let $\mathcal{X}$ be a locally convex topological vector space equipped with a partial order~$\leq$.
Let $\mathcal{X}_+ = \{x \in \mathcal{X}\,|\, x \geq 0\}$ be the positive cone. 
We denote the topological dual of $\mathcal{X}$ by $\mathcal{X}^* = \{ x^*: \mathcal{X} \rightarrow \mathbb{R} \,|\, x^* \text{ is linear, continuous}\}$ and let $\mathcal{X}^*_+ = \{x^* \in \mathcal{X}^* \,|\, \forall x \in \mathcal{X}_+ : x^*(x) \geq 0 \}$ be the dual cone of $\mathcal{X}_+$. 

The key result for duality is the following proposition which yields an alternative representation of the acceptance set by a set of probability measures.

\begin{proposition} \label{Lemma:EquivOfXinAcc}
Let $\acc \subset \mathcal{X}$ be a $\text{weak}^*$-closed, convex acceptance set.
Then $x \in \acc$ if and only if for all $x^* \in \mathcal{X}^*_+$ the inequality $\alpha(x^*) := \inf_{y \in \acc} x^*(y) \leq x^*(x)$ holds.
\end{proposition}
\begin{proof}
The `only if' implication follows directly, since the infimum of any functional $x^*$ over $\acc$ is always less or equal than the value $x^*(x)$ for some arbitrary $x \in \acc$.
For the `if' direction we use a version of the Hahn-Banach Separation Theorem on locally convex topological vector spaces, see for example Theorem V.2.10 in N.~Dunford and J.~T.~Schwartz \cite{bib:DS}.
It yields for any $x \in \mathcal{X} \setminus \acc$ a linear functional $\ell \in \mathcal{X}^*$ such that $\ell(-x) > \sup_{y \in \acc} \ell(-y)$.
By linearity of $\ell$, $\inf_{y \in \acc} \ell(y) > \ell(x)$ follows.
To show that $\ell \in \mathcal{X}^*_+$, take some $y \in \acc$, then monotonicity of $\acc$ implies that for any $z_+ \in \mathcal{X}_+$ also $y + z_+ \in \acc$.
So linearity of $\ell$ implies that $\ell(y) + \ell(z_+) > \ell(x)$.
But this can only be true if $\ell(z_+) \geq 0$ for all $z_+ \in \mathcal{X}_+$ which means $\ell \in \mathcal{X}^*_+$. 
Indeed, assume $\ell(z_+) < 0$ for some $z_+ \in \mathcal{X}_+$.
Since $\lambda z_+ \in \mathcal{X}_+$, for any positive real $\lambda$, we get $\ell(\lambda z_+) = \lambda \ell(z_+) \rightarrow -\infty$ as $\lambda \rightarrow \infty$, leading to a contradiction. 
\end{proof}

In a slightly different setting, the proof can be found in Lemma C.3 of~\cite{bib:Drapeau}.

\begin{remark} \label{Remark:DualsOfLinfty}
Choosing the $\text{weak}^*$-topology on $L^\infty(\mathbb{P})$ denoted by $\sigma(L^{\infty},L^1)$, the dual of $L^{\infty}(\mathbb{P})$ is $L^1(\mathbb{P})$, see for example Theorem V.3.9 in \cite{bib:DS}.
A short calculation shows that each $f \in L^1_+(\mathbb{P})$ defines a $\sigma$-additive measure $\mathbb{Q} \ll \mathbb{P}$ via the integral 
\begin{align} \label{Eq:DefineIntegralL1}
\mathbb{Q}[A] := \frac{1}{\Vert f \Vert_{L^1(\mathbb{P})}} \int_{\Omega} f \mathbf{1}_A \,\mathrm{d}\mathbb{P} \,,\; \text{ for } A \in \mathcal{F}\,,
\end{align}
such that $\mathbb{E}_{\mathbb{Q}}[g] \geq 0$, for all $g \in L^{\infty}_+(\mathbb{P})$.
Using the Radon-Nikod\'ym Theorem, as for example stated in \cite[Theorem III.10.2]{bib:DS}, we can verify the other direction that for each \mbox{$\mathbb{Q} \in \mathcal{M}_{\sigma}(\mathbb{P})$} there exists an $f \in L^1_+(\mathbb{P})$ such that \hyperref[Eq:DefineIntegralL1]{Equation (\ref*{Eq:DefineIntegralL1})} holds for all $A \in \mathcal{F}$.
\end{remark}

The following lemma shows that the expectation with respect to probability measures $\mathbb{Q} \in \mathcal{M}_{\sigma}(\mathbb{P})$ can be used to represent our acceptance sets in $L^\infty(\mathbb{P})$.

\begin{lemma} \label{Lemma:XinADual}
Let $\acc \subset \mathcal{X} =  L^{\infty}(\Omega,\mathcal{F},\mathbb{P})$ be a $\sigma(L^{\infty},L^1)$-closed, convex acceptance set.
Then $X_T \in \acc$ if and only if for all probability measures $\mathbb{Q} \in \mathcal{M}_{\sigma}(\mathbb{P})$
\begin{align*}
 \inf_{Y_T \in \acc} \mathbb{E}_{\mathbb{Q}}[Y_T] \leq \mathbb{E}_{\mathbb{Q}}[X_T]\,.
\end{align*}
\end{lemma}
\begin{proof}
\hyperref[Lemma:EquivOfXinAcc]{Proposition \ref*{Lemma:EquivOfXinAcc}} with $\mathcal{X} = L^{\infty}(\Omega,\mathcal{F},\mathbb{P})$ and \hyperref[Remark:DualsOfLinfty]{Remark \ref*{Remark:DualsOfLinfty}} yield the assertion. 
\end{proof}

Using this result we can now derive a dual representation for intrinsic risk measures.

\begin{theorem}[Dual representation] \label{Thm:DualRepresofIntrinsicRisk}
Let $\acc \subset \mathcal{X} = L^\infty(\Omega,\mathcal{F},\mathbb{P})$ be a $\sigma(L^{\infty},L^1)$-closed, convex acceptance set containing $0$ and let $S$ be an eligible asset.
For $\mathbb{Q} \in \mathcal{M}_{\sigma}(\mathbb{P})$ define $\alpha(\mathbb{Q},\acc) = \inf_{X_T \in \acc} \mathbb{E}_{\mathbb{Q}}[X_T]$.
Then the intrinsic risk measure takes the form 
\begin{align} \label{Eq:ConvexRepresentationIntrinsicRisk}
\riskAR{\acc}{S}{X} = \sup_{\mathbb{Q} \in \mathcal{M}_{\sigma}(\mathbb{P})} \frac{(\alpha(\mathbb{Q},\acc) - \mathbb{E}_{\mathbb{Q}}[X_T])^+}{ \frac{X_0}{S_0}\mathbb{E}_{\mathbb{Q}}[S_T] - \mathbb{E}_{\mathbb{Q}}[X_T]} \,.
\end{align}
\end{theorem}
\begin{proof}
With help of \hyperref[Lemma:XinADual]{Lemma \ref*{Lemma:XinADual}} we can rewrite the defining equation of the intrinsic risk measure as
\begin{align*}
\riskR{A}{S}{X} &= \inf \Big\{\lambda \in [0,1] \,|\, (1-\lambda)X_T + \lambda \tfrac{X_0}{S_0} S_T \in \acc \Big\} \\
 &= \inf \Big\{\lambda \in [0,1] \,|\, \forall \mathbb{Q} \in \mathcal{M}_{\sigma}(\mathbb{P}) :\, \mathbb{E}_\mathbb{Q}\big[(1-\lambda)X_T + \lambda \tfrac{X_0}{S_0} S_T \big] \geq \alpha(\mathbb{Q},\acc) \Big\} \\
 &= \inf \Big\{\lambda \in [0,1] \,|\, \forall \mathbb{Q} \in \mathcal{M}_{\sigma}(\mathbb{P}) :\, \lambda  \mathbb{E}_\mathbb{Q}\big[\tfrac{X_0}{S_0} S_T - X_T \big]  \geq \alpha(\mathbb{Q},\acc) - \mathbb{E}_\mathbb{Q}[X_T] \Big\} \,.
\end{align*}
Note that if $X_T \in \acc$, then \hyperref[Lemma:XinADual]{Lemma \ref*{Lemma:XinADual}} implies that for all $\mathbb{Q} \in \mathcal{M}_{\sigma}(\mathbb{P})$ the expression $\alpha(\mathbb{Q},\acc) - \mathbb{E}_\mathbb{Q}[X_T]$ is negative and thus, the infimum over $\lambda$ is equal to $0$.
Assume now that $X_T \notin \acc$, then $\mathbb{E}_\mathbb{Q}[\frac{X_0}{S_0} S_T] - \mathbb{E}_\mathbb{Q}[X_T] \geq \alpha(\mathbb{Q},\acc) - \mathbb{E}_\mathbb{Q}[X_T] > 0$ and we can rearrange to get
\begin{align*}
\riskR{A}{S}{X} &= \inf \Bigg\{\lambda \in [0,1] \,\Big|\, \forall \mathbb{Q} \in \mathcal{M}_{\sigma}(\mathbb{P}) :\, \lambda \geq  \frac{\alpha(\mathbb{Q},\acc) - \mathbb{E}_\mathbb{Q}[X_T]}{\frac{X_0}{S_0} \mathbb{E}_\mathbb{Q}[ S_T] - \mathbb{E}_\mathbb{Q}[X_T]} \Bigg\} \\
&= \inf \Bigg\{\lambda \in [0,1] \,\Big|\, \, \lambda \geq \sup_{\mathbb{Q} \in \mathcal{M}_{\sigma}(\mathbb{P})} \frac{\alpha(\mathbb{Q},\acc) - \mathbb{E}_\mathbb{Q}[X_T]}{\frac{X_0}{S_0} \mathbb{E}_\mathbb{Q}[ S_T] - \mathbb{E}_\mathbb{Q}[X_T]} \Bigg\} \\
&= \sup_{\mathbb{Q} \in \mathcal{M}_{\sigma}(\mathbb{P})} \frac{\alpha(\mathbb{Q},\acc) - \mathbb{E}_\mathbb{Q}[X_T]}{\frac{X_0}{S_0} \mathbb{E}_\mathbb{Q}[ S_T] - \mathbb{E}_\mathbb{Q}[X_T]} \,.
\end{align*}
From here the representation in \eqref{Eq:ConvexRepresentationIntrinsicRisk} follows. 
\end{proof}
\noindent Also for this result recall \hyperref[Rem:HebbareSingularity]{Remark \ref*{Rem:HebbareSingularity}} for well-definedness.

{It is interesting to compare this representation to the representation of convex risk measure given in \hyperref[Eq:ConvexEquivalencesOnProbSpace]{Equation (\ref*{Eq:ConvexEquivalencesOnProbSpace})}.
We notice that the numerator in \hyperref[Eq:ConvexRepresentationIntrinsicRisk]{Equation (\ref*{Eq:ConvexRepresentationIntrinsicRisk})} contains the same terms as the expression in \hyperref[Eq:ConvexEquivalencesOnProbSpace]{Equation (\ref*{Eq:ConvexEquivalencesOnProbSpace})}, since $\alpha(\mathbb{Q},\acc) = -\alpha_\mathrm{min}(\mathbb{Q},\acc)$.
However, before the supremum is taken over $\mathcal{M}_{\sigma}(\mathbb{P})$ the numerator is normalised by an expected distance of financial position and eligible asset.
}

{If the acceptance set is a cone, then we can even link} \hyperref[Thm:DualRepresofIntrinsicRisk]{Theorem \ref*{Thm:DualRepresofIntrinsicRisk}} via the dual representation of coherent risk measures in \hyperref[Thm:DualRepresOfCoherentRisk]{Corollary \ref*{Thm:DualRepresOfCoherentRisk}} to \hyperref[Prop:IntrinsicRepresPosHom&S-Additive]{Theorem \ref*{Prop:IntrinsicRepresPosHom&S-Additive}}.

\begin{corollary} \label{Cor:DualRepresIntrinsicCoherent}
If $\acc$ is a $\sigma(L^{\infty},L^1)$-closed, convex cone and $\frac{S_T}{S_0} = \mathbf{1}_\Omega$, then we recover the representation 
\begin{align*}
\riskAR{\acc}{S}{X} = \frac{(\risk{A}{S}{X_T})^+}{X_0 + \risk{A}{S}{X_T}} \,.
\end{align*}
\end{corollary}
\begin{proof}
Since $\acc$ is a cone, the (minimal) penalty function in \hyperref[Thm:DualRepresofIntrinsicRisk]{Theorem \ref*{Thm:DualRepresofIntrinsicRisk}} can only take the values $0$ and $\pm\infty$ as for all $\mathbb{Q} \in \mathcal{M}_{\sigma}(\mathbb{P})$ and all $\lambda > 0$ the following equality holds,
\begin{align*}
\alpha(\mathbb{Q},\acc) = \sup_{Y_T \in \acc} \mathbb{E}_{\mathbb{Q}}[-Y_T] = \sup_{X_T \in \lambda \acc} \mathbb{E}_{\mathbb{Q}}[-X_T] = \sup_{Y_T \in \acc} \mathbb{E}_{\mathbb{Q}}[-\lambda Y_T] = \lambda \alpha(\mathbb{Q},\acc)\,.
\end{align*}
Restricting $\mathcal{M}_{\sigma}(\mathbb{P})$  to $\mathcal{M} =  \{\mathbb{Q} \in \mathcal{M}_{\sigma}(\mathbb{P}) \,|\, \alpha(\mathbb{Q},\acc) = 0 \}$, \hyperref[Thm:DualRepresofIntrinsicRisk]{Theorem \ref*{Thm:DualRepresofIntrinsicRisk}} yields the representation
\begin{align*}
\riskR{A}{S}{X} = \sup_{\mathbb{Q} \in \mathcal{M}} \frac{(\mathbb{E}_\mathbb{Q}[-X_T])^+}{X_0 + \mathbb{E}_\mathbb{Q}[-X_T]} \,.
\end{align*}
Since for any constant $c > 0$ the map $x \mapsto \frac{x}{c + x}$ is increasing on $\mathbb{R}_{>0}$, we can split the supremum to get
\begin{align*}
\riskR{A}{S}{X} = \frac{\sup_{\mathbb{Q} \in \mathcal{M}}  (\mathbb{E}_\mathbb{Q}[-X_T])^+}{  X_0 + \sup_{\mathbb{Q} \in \mathcal{M}}  \mathbb{E}_\mathbb{Q}[-X_T]} \,.
\end{align*}
But $\sup_{\mathbb{Q} \in \mathcal{M}}  \mathbb{E}_\mathbb{Q}[-X_T]$ is the representation of coherent risk measures from \hyperref[Thm:DualRepresOfCoherentRisk]{Corollary \ref*{Thm:DualRepresOfCoherentRisk}}, and thus, we recover the form
\begin{align*}
\riskR{A}{S}{X} = \frac{(\rho_{\acc,S}(X_T))^+}{X_0 + \rho_{\acc,S}(X_T)} \,,
\end{align*}
the representation of intrinsic risk measures with respect to conic acceptance sets derived in \hyperref[Sec:IntrinsicRiskonPosHom&Cashadditive]{Section \ref*{Sec:IntrinsicRiskonPosHom&Cashadditive}}. 
\end{proof}

\section{Conclusion} \label{Sec:outlook}
In this paper we have developed a new type of risk measure: intrinsic risk measures.
We argued that since traditional risk measures are defined via hypothetical additional capital, which is not always available in reality, it is natural to consider risk measures that only allow the usage of internal capital contained in the financial position.

We discussed basic properties, provided examples and applications, an alternative representation on conic acceptance sets with a direct comparison to traditional risk measures and their efficiency, and a dual representation on convex acceptance sets.
This new concept allows us to extend the scope of applications and eliminate problems with infinite values, while concentrating on the primary objective to reach acceptability starting from an unacceptable position.

We have shown that standard properties such as monotonicity and quasi-convexity are imposed directly through the structure of the acceptance set as opposed to monetary risk measures.
On conic acceptance sets, such as the ones associated with Value at Risk and Expected Shortfall, we have drawn connections between intrinsic risk measures and their traditional counterparts.
We showed that the former approach requires less capital to reach acceptability and at the same time yields financial positions with the same performance.
As the representation on cones cannot be extended to convex acceptance sets, we used duality results for convex sets and derived a dual representation for the intrinsic risk measure in terms of $\sigma$-additive probability measures.\\

At the end we would like to mention some ideas for the further development of the research on intrinsic risk measures.
For simplicity reasons we chose $\mathcal{X} = L^\infty(\Omega,\mathcal{F},\mathbb{P})$ and we indicated where immediate extensions to more general spaces are possible.
Given the importance of these spaces in mathematical finance it is of interest to extend intrinsic risk measures to general ordered topological vector spaces, as it was done in \cite{bib:WPC2} for traditional risk measures. 
As some of our results require that the interior of the acceptance set is non-empty, the interior would have to be substituted by a refined concept such as the core.
Moreover, extensions to multiple financial positions and multiple eligible assets could be considered.
Since the intrinsic risk measure operates on a single financial position~$X$, portfolios need to be aggregated before the risk measure can be applied.
A possibly more sensible approach would treat the positions that constitute the portfolio individually, but dependent on their weights in the portfolio. 
Then one has to decide how to define the infimum over all possible risk vectors.
Yet another approach could be to define multidimensional acceptance sets without aggregating the positions.

\newpage
\bibliographystyle{alpha}
\bibliography{bibliography}

\begin{thebibliography}{CVMMM11}

\bibitem[ADEH99]{bib:ADHE}
Philippe Artzner, Freddy Delbaen, Jean-Marc Eber, and David Heath.
\newblock Coherent measures of risk.
\newblock {\em Mathematical Finance}, 9(3):203--228, 1999.

\bibitem[ADKM09]{bib:ADKM}
Philippe Artzner, Freddy Delbaen, and Pablo Koch-Medina.
\newblock Risk measures and efficient use of capital.
\newblock {\em Astin Bulletin}, 39(1):101--116, 2009.

\bibitem[CVMMM11]{bib:CMMM}
Simone Cerreia-Vioglio, Fabio Maccheroni, Massimo Marinacci, and Luigi
  Montrucchio.
\newblock Risk measures: Rationality and diversification.
\newblock {\em Mathematical Finance}, 21(4):743--774, 2011.

\bibitem[Del02]{bib:Delbaen2002}
Freddy Delbaen.
\newblock Coherent risk measures on general probability spaces.
\newblock In {\em Advances in Finance and Stochastics: Essays in Honour of
  Dieter Sondermann}. Springer Berlin Heidelberg, Berlin, Heidelberg, 2002.

\bibitem[DK13]{bib:Drapeau}
Samuel Drapeau and Michael Kupper.
\newblock Risk preferences and their robust representation.
\newblock {\em Mathematics of Operations Research}, 38(1):28--62, 2013.

\bibitem[DS58]{bib:DS}
Nelson Dunford and Jacob~T. Schwartz.
\newblock {\em Linear Operators. Part I: General Theory}.
\newblock Interscience Publishers Inc., New York, 1958.

\bibitem[EKR09]{bib:EKR}
Nicole El~Karoui and Claudia Ravanelli.
\newblock Cash subadditive risk measures and interest rate ambiguity.
\newblock {\em Mathematical Finance}, 19(4):561--590, 2009.

\bibitem[FKMM14a]{bib:WPC2}
Walter Farkas, Pablo Koch-Medina, and Cosimo Munari.
\newblock Beyond cash-additive risk measures: when changing the num{\'e}raire
  fails.
\newblock {\em Finance and Stochastics}, 18(1):145--173, 2014.

\bibitem[FKMM14b]{bib:WPC1}
Walter Farkas, Pablo Koch-Medina, and Cosimo Munari.
\newblock Capital requirements with defaultable securities.
\newblock {\em Insurance: Mathematics and Economics}, 55:58 -- 67, 2014.

\bibitem[FS04]{bib:FS}
Hans F{\"o}llmer and Alexander Schied.
\newblock {\em Stochastic Finance: An Introduction in Discrete Time}.
\newblock De Gruyter studies in mathematics. Walter de Gruyter, 2004.

\bibitem[FS06]{bib:FrittelliScandolo}
Marco Frittelli and Giacomo Scandolo.
\newblock Risk measures and capital requirements for processes.
\newblock {\em Mathematical Finance}, 16(4):589--612, 2006.

\bibitem[KK11]{bib:KK}
Dimitrios~G. Konstantinides and Christos~E. Kountzakis.
\newblock Risk measures in ordered normed linear spaces with non-empty
  cone-interior.
\newblock {\em Insurance: Mathematics and Economics}, 48(1):111--122, 2011.

\bibitem[Mun15]{bib:Munari}
Cosimo-Andrea Munari.
\newblock {\em Measuring risk beyond the cash-additive paradigm}.
\newblock PhD thesis, Diss.~ETH no.~22541, ETH Z\"urich, 2015.

\bibitem[Smi16]{bib:AS}
Alexander Smirnow.
\newblock Risk measures: recent developments and new ideas.
\newblock Master's thesis, Universit\"at Z\"urich, ETH Z\"urich, 2016.

\end{thebibliography}

\end{document}